\documentclass[12pt,notitlepage]{article}%
\usepackage[top=1.5in, bottom=1.5in, left=1.5in, right=1.5in]{geometry}
\usepackage{amsthm}
\usepackage{bbm}
\usepackage{blkarray}
\usepackage{CJK}
\usepackage{nicematrix}
\usepackage{setspace}
\usepackage{dsfont}
\usepackage{eurosym}
\usepackage{subfigure}
\usepackage{xcolor}
\usepackage{amssymb}
\usepackage{url}
\usepackage{supertabular}
\usepackage{graphicx}
\usepackage{float}
\usepackage[colorlinks,linkcolor=links,citecolor=cites,urlcolor=MyDarkBlue]%
{hyperref}
\usepackage{amsfonts}
\usepackage{enumitem}
\usepackage{amsmath}
\usepackage{amstext}
\usepackage{appendix}
\usepackage{tikz}
\usepackage{rotating}
\usepackage{lscape}
\usepackage[longnamesfirst]{natbib}
\usepackage{chngpage}
\usepackage{esint}
\usepackage{bm}
\usepackage{mathrsfs}
\providecommand{\U}[1]{\protect\rule{.1in}{.1in}}
\newtheorem{theorem}{Theorem}

\newtheorem{assumption}{Assumption}

\newtheorem{condition}{Condition}

\newtheorem{corollary}{Corollary}
 
\newtheorem{example}{Example}

\newtheorem{lemma}{Lemma}

\newtheorem{proposition}{Proposition}
\newtheorem{remark}{Remark}

\newtheorem{distributionfamily}{Distribution family}
\newcommand{\RNum}[1]{\uppercase\expandafter{\romannumeral #1\relax}}
\DeclareMathOperator*{\argmax}{argmax}

\DeclareMathOperator*{\argmin}{argmin}

\usepackage{indentfirst}
\definecolor{MyDarkBlue}{rgb}{0,0.08,0.45}
\definecolor{cites}{HTML}{324b13}
\definecolor{links}{HTML}{1a663b}
\definecolor{MyLightMbuyera}{cmyk}{0.1,0.8,0,0.1}
\hypersetup{
colorlinks,citecolor=blue,filecolor=black,linkcolor=blue,urlcolor=blue
}
\usepackage{chngcntr}
\counterwithout{equation}{section}
\usepackage{titling}
\begin{document}
\title{Optimally Selling Heterogeneous Objects by Quantity\thanks{I am indebted to my supervisor Charles Zheng and committee members  Nail Kashaev and Bruno Salcedo   for their invaluable guidance and encouragement.   I also thank Yeon-Koo Che,  Yijia Ding,  Maria Goltsman, Kevin He, Todd Kaplan,  David Lagziel, Robert McCann, Andrew McClellan, Fedor
Sandomirskiy,   Dazhong Wang, Hanzhe Zhang,   and especially Tim Conley and Greg Pavlov, for helpful discussions and comments throughout the development of the project.   All remaining errors are mine. 
}}
\author{{Xun Chen}%
\thanks{Department of Economics, the University of Western Ontario. E-mail: xche225@uwo.ca
}}

\maketitle

 \begin{abstract}
I study the optimality of selling $K$ heterogeneous objects to a buyer with private values through a quantity-based menu: the seller posts  a price for  each number of objects, and the buyer chooses a quantity, pays its price, and selects any objects up  to that quantity. I propose a pricing method that aggregates critical values for second differences of quantity  values---the highest values attainable at each quantity. Critical values can be determined separately when each second difference's conditional inverse-hazard ratio depends only on that difference. For additive values, I show that the resulting menu is optimal among all incentive-compatible and individually rational mechanisms under explicit distributional and support conditions, including ex ante symmetry  across object labels. A separate result applies the same pricing method to obtain an optimal menu for nonadditive bundle values without ex ante symmetry, assuming free disposal, independence between the quantity-value vector and the vector of bundle shortfalls from their respective quantity values, and corresponding distributional and support conditions. For an arbitrary number of objects, applications yield explicit optimal menus  that involve quantity discounts, sparse quantity options, or two-part tariffs. These results differ from those establishing the optimality of pure bundling in \cite{HaphpanahHartline2021}, nested bundling in \cite{10.1257/aer.20230455}, separate sales in \cite{Kwak2026}, and core-peripheral bundling in \cite{Zhong2026}. 

\vspace{0.15cm}%a

\end{abstract}

\section{Introduction}

\begin{spacing}{1.16}

A seller wishes to sell several heterogeneous objects to a buyer whose values are private. The seller only knows the distribution of values and seeks a mechanism that maximizes expected revenue subject to incentive compatibility and individual rationality constraints. With one object, an optimal mechanism posts a take-it-or-leave-it price \citep{RileyZEckhauser1983}. With multiple objects, exact solutions are known for particular settings, including i.i.d. uniform values for up to six objects \citep{giannakopoulos_koutsoupias_2018_sja} and environments in which pure bundling, separate sales, nested bundling, or core-peripheral bundling is optimal \citep{Giannakopoulos2015,HaphpanahHartline2021,Kwak2026,10.1257/aer.20230455,Zhong2026}. The general problem remains difficult. Optimal revenue need not increase when the buyer's value distribution shifts upward \citep{hart2015maximal}, randomization may be necessary \citep{THANASSOULIS2004217,MANELLI20061,MANELLI2007153,Pavlov2011}, and an optimal mechanism may require infinitely many menu options \citep{DaskalakisDeckelbaumTzamos2017,HART2019991}.

This study develops a  pricing method and establishes conditions under which the resulting menu maximizes expected revenue among all feasible incentive-compatible and individually rational mechanisms. The  method is based on second differences of quantity values---the highest values attainable at each quantity. Assuming each second difference's conditional inverse-hazard ratio depends only on that difference, I determine critical values separately by setting  the corresponding virtual-surplus functions to zero and obtain quantity prices through two successive summations.  This pricing method applies to both additive object values and nonadditive bundle values.   
 The resulting prices form a quantity-based menu: the seller sets a price for each number of objects; the buyer chooses a quantity, pays the corresponding price, and selects any objects up to that quantity  \citep[see, e.g.,][]{HittChen2005}. For example, a television provider can charge by the number of channels that the consumer selects from an eligible list.\footnote{See, e.g.,  the Pick Packs offered  by WTC Communications in Canada: \url{https://www.wtccommunications.ca/for-home/television/channel-lineup}.} Quantity-based menus also appear in city attraction passes,\footnote{See, e.g., Go City's New York Explorer Pass, which is valid for one person and permits only one visit to each selected attraction: \url{https://gocity.com/en/new-york/passes/explorer}. For these usage restrictions, see \url{https://gocity.com/en/new-york/terms}.} cinema passes,\footnote{See, e.g., the Joint Film Pass offered by Lighthouse in Poole, United Kingdom, which includes five pairs of tickets for standard screenings of the holder's choice, with at  most one pair redeemable per screening: \url{https://www.lighthousepoole.co.uk/visit-us/film-pass/}.}  and video game bundles.\footnote{See, e.g., Fanatical's build-your-own Retro Revival Bundle: \url{https://www.fanatical.com/en/pick-and-mix/build-your-own-retro-revival-bundle}.}

 %certain financial derivatives,\footnote{See, e.g., ANZ's Flexi Cap, which allows buyers to purchase a specified number of interest-rate protection rights and subsequently choose the periods in which to exercise  them: \url{https://www.anz.com/Documents/FXOnline/2003_12_19_Flexi\%20Cap.pdf}. Exercise decisions occur sequentially over time.}

My first main result establishes optimality of the resulting menu among all feasible incentive-compatible and individually rational mechanisms for additive values (Theorem \ref{EndogenousGapFlowGap} and Proposition \ref{ClosedFormExactOptimality}). Under additivity, the first differences of quantity values are ranked object values, and the second differences are gaps between adjacent ranked values.\footnote{These gaps also appear in \citet[proof of Proposition~2]{BikhchandaniMishra2024}.} Ex ante symmetry (Assumption \ref{ex}) permits the rank-allocation reduction established by \citet{BikhchandaniMishra2024}, while gap rent regularity (Assumption \ref{MightAsWelllWalkingOnTheSun}) allows separate critical pricing. The density and support restrictions ensure that quantities maximizing value net of accumulated critical values also maximize quantity virtual surplus (Assumption \ref{TopRankSignAlignment}) and that information rent becomes zero at the relevant lower boundaries (Assumption \ref{YouDontDeserveIt}). The menu therefore attains the revenue upper bound. The state-dependent extension provides a sufficient test for the optimality of a proposed quantity-based menu under a given value distribution (Proposition \ref{StateDependentExactOptimality}).  The results also extend to multiple units with decreasing private marginal values (Remarks~\ref{MultipleUnitsRevenueEquivalence} and~\ref{HomogeneousUnitsWithoutExchangeability}).

My second main result establishes the optimality of the same pricing construction for nonadditive bundle values without ex ante symmetry (Proposition \ref{NonadditiveQuantityVerification}). For bundle values, the quantity value at $n$ is the largest order statistic among size-$n$ bundle values.  I assume free disposal, the bundle-value counterparts of Assumptions \ref{MightAsWelllWalkingOnTheSun}--\ref{YouDontDeserveIt}, and independence between the vector of highest bundle values across quantities and the vector of within-quantity shortfalls (Assumption \ref{BundleValueIndependence}). I derive the revenue upper bound directly for all feasible incentive-compatible and individually rational mechanisms. Under a quantity-based menu, the buyer's self-selection eliminates the shortfall term, and the corresponding verification conditions ensure that the menu attains the bound. For $K\geq2$, the independence condition excludes full-dimensional additive distributions, so this result does not subsume the additive baseline. An extension retains separate critical pricing provided each active second difference's conditional inverse-hazard ratio is invariant to the other second differences and all shortfalls (Assumption \ref{BundleConditionalInverseHazardInvariance} and Corollary \ref{ConditionalBundleCriticalPricing}).

For additive values, I derive explicit optimal menus for any number of objects in four families of ex ante symmetric distributions. Uniform and Gaussian gap specifications with joint support restrictions yield menus offering only $K-1$ objects or the full bundle (Distribution Family \ref{GaussianGapDensityFamily}). Exponential initial segments in one gap density yield pure bundling, a common price per object, or sparse menus (Distribution Family \ref{SingleActiveGapDensityFamily}). Independent reverse-Beta gaps yield quantity discounts (Distribution Family \ref{FullActiveGapDensityFamily}). Exponential initial segments in the first and terminal gap densities, together with joint support restrictions, yield two-part tariffs (Distribution Family \ref{TwoActiveGapDensityFamily}). The induced object values can be correlated. For two objects with nonadditive values, I also derive optimal quantity prices under a jointly uniform bundle-value distribution on a restricted support (Example~\ref{UniformBundleFreeDisposal}). 

The following three two-object examples specify joint densities directly in the original object-value space. Example~\ref{JointUniformLeadingExample} is the uniform specification of Distribution Family~\ref{GaussianGapDensityFamily}.\footnote{For a related uniform ordered-value example, see \citet[p.~589]{Armstrong2016}.} Example~\ref{UnboundedLeadingExample} illustrates the unbounded extension of Distribution Family~\ref{SingleActiveGapDensityFamily}. Example~\ref{StateDependentNonpurchaseExample} specifies a distribution with full support on the unit square that may describe specialized interests.\footnote{For example, \citet[Study~2]{BellMarshall2003} find that consumers more involved in food preparation and eating show larger differences in liking ratings across variants of the same food.} This example lies outside all four distribution families; optimality follows from the state-dependent verification in Proposition~\ref{StateDependentExactOptimality}.

\par\smallskip\noindent\begin{minipage}{\linewidth}
\setstretch{1.16}
\begin{example}[Jointly uniform object values with support restrictions]\label{JointUniformLeadingExample}
Let $K=2$. The joint density is 
$$
f(\theta)=\frac{2}{17}\mathbf1_\Theta(\theta),\;\;\Theta=\left\{\theta\in[1,6]^2:
\begin{array}{l}
|\theta_1-\theta_2|\leq2,\; \min\{\theta_1,\theta_2\}\leq4,\\
\max\{\theta_1-3,\theta_2-3,\theta_1+\theta_2-5\}\geq0
\end{array}\right\}.
$$ The optimal menu charges $3$ and $5$ for quantities $1$ and $2$, respectively. 
\end{example}

\begin{example}[Exponential object values with bounded differences]\label{UnboundedLeadingExample}
Let $K=2$. Draw two independent $\mathrm{Exp}(1)$ values and condition on the support restriction $|\theta_1-\theta_2|\leq1/2$. The resulting joint density is
$$
f(\theta)=\frac{e^{-(\theta_1+\theta_2)}}{1-e^{-1/2}}\mathbf1_\Theta(\theta),\;\;
\Theta=\{\theta\in\mathbb R_+^2:|\theta_1-\theta_2|\leq1/2\}.
$$ The optimal menu charges $1/2$ and $1$ for quantities $1$ and $2$, respectively.
\end{example}
\begin{figure}[H]
\setstretch{1}
\centering
\includegraphics[width=0.85\textwidth]{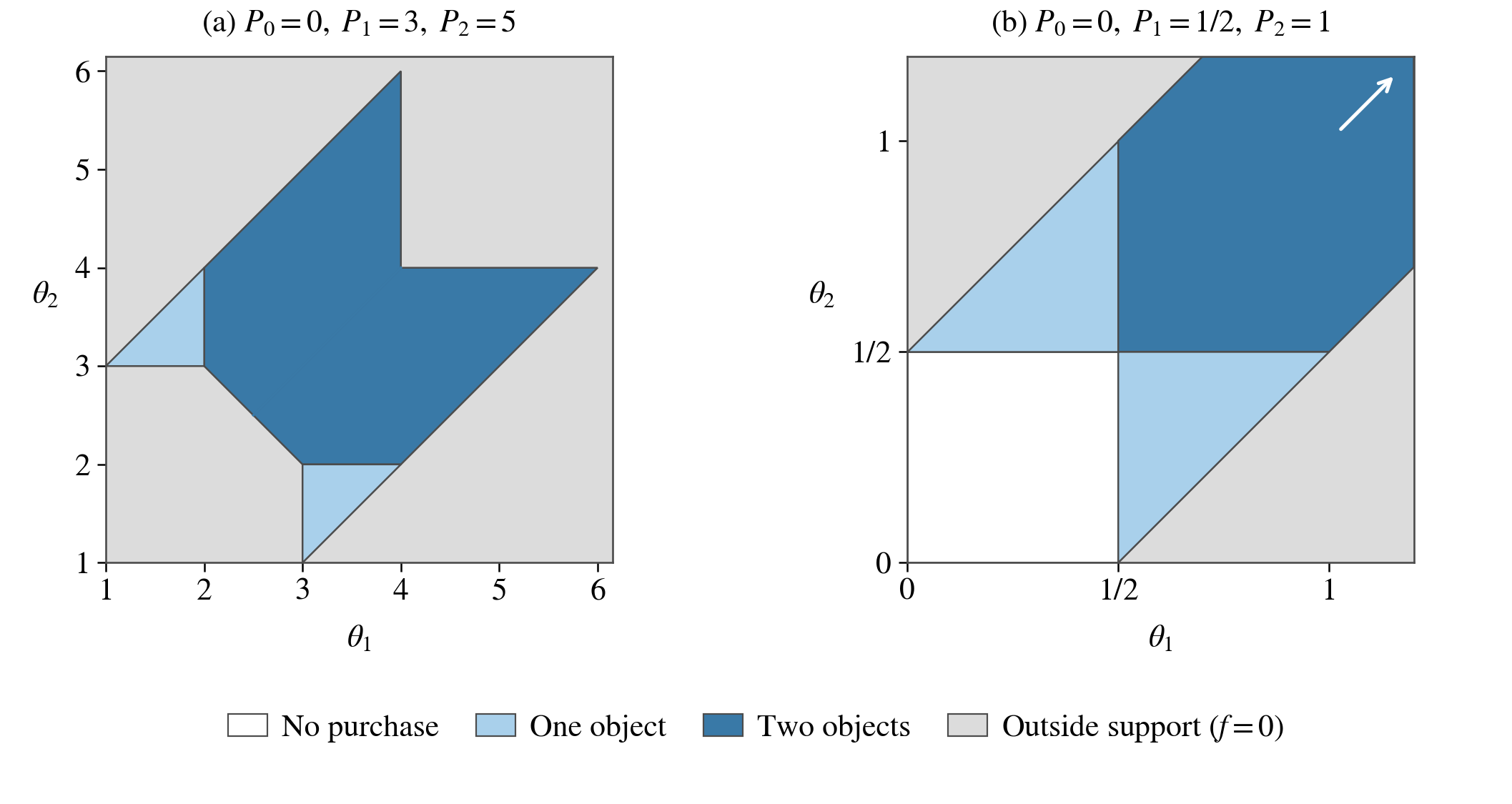}
\caption{Optimal quantity choices in Examples~\ref{JointUniformLeadingExample} and \ref{UnboundedLeadingExample}, shown in panels (a) and (b), respectively. The arrow in panel (b) indicates that the support extends to infinity.}
\label{JointUniformLeadingFigure}
\label{UnboundedLeadingFigure}
\end{figure}
\end{minipage}\par\smallskip
\par\smallskip\noindent\begin{minipage}{\linewidth}
\setstretch{1.16}
\begin{example}[Object values with a  centrally antisymmetric density perturbation]\label{StateDependentNonpurchaseExample}
Let $K=2$ and $\Theta=[0,1]^2$. The joint density is
$$
f(\theta)=\mathbf1_\Theta(\theta)
\begin{cases}
1-|\theta_1-\theta_2|,&\theta_1+\theta_2<1,\\
1+|\theta_1-\theta_2|,&\theta_1+\theta_2\geq1.
\end{cases}
$$
This exchangeable distribution has full support on $\Theta$. For a given total value below $1$, the density decreases with $|\theta_1-\theta_2|$; for a total value at least $1$, it increases with this difference. This distribution can describe values for goods associated with specialized interests: a buyer with little interest may value both goods similarly and modestly, whereas an enthusiast may value them more highly while expressing a stronger preference for one over the other. The optimal menu offers only quantity $2$ at price $1$.
\end{example}

\begin{figure}[H]
\setstretch{1}
\centering
\includegraphics[width=0.85\textwidth]{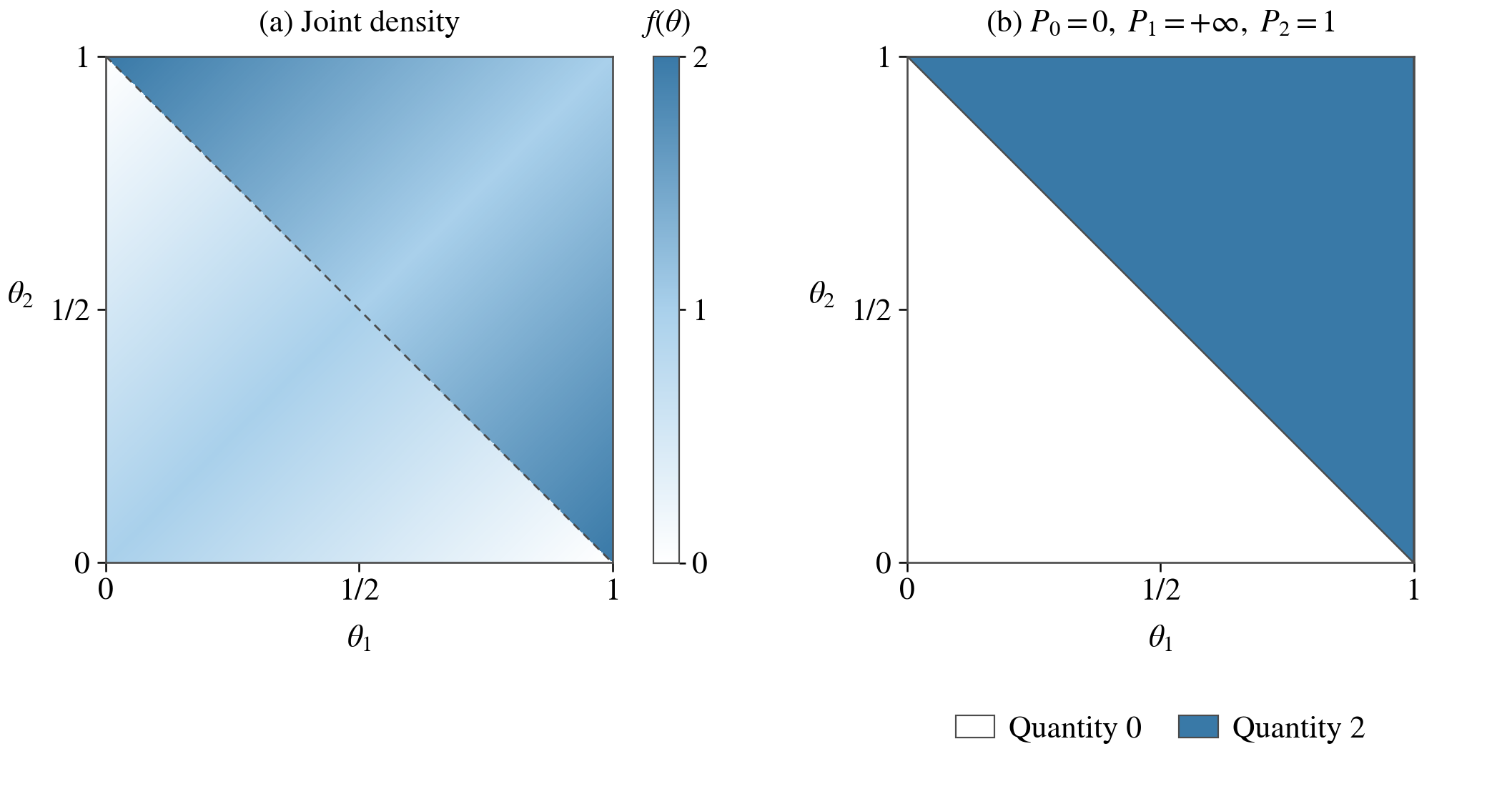}
\caption{Joint density and optimal quantity choices in Example~\ref{StateDependentNonpurchaseExample}. Panel (a) shows the joint density, with darker shades indicating higher density; the dashed line is $\theta_1+\theta_2=1$. Panel (b) shows the optimal purchase regions: the buyer chooses quantity $2$ when $\theta_1+\theta_2>1$ and quantity $0$ otherwise.}
\label{StateDependentDensityHeatmap}
\label{StateDependentAllocationsFigure}
\end{figure}
\end{minipage}\par\smallskip

\noindent\textbf{Related literature.} My analysis contributes to revenue-decomposition methods in multidimensional mechanism design \citep{Armstrong1996,HaphpanahHartline2021,Kwak2026}. It also relates to \cite{Pavlov2011}, who adopts valuation differences to study optimal mechanisms for two objects. My revenue decomposition is based on second differences of quantity values in both the additive and nonadditive environments; under additivity, these coordinates coincide with gaps between ranked object values. 

My study also relates to quantity-based menus, also referred to as bundle-size pricing or cardinality bundling. These menus restrict deterministic mixed bundling to prices that depend only on bundle size.  \cite{HittChen2005} establish conditions under which customized bundling reproduces an optimal deterministic mixed-bundling solution. \cite{Chuetal2011} show that bundle-size pricing often closely approximates mixed-bundling profits. \cite{Wuetal2019} derive optimal cardinality prices under a single-crossing condition, while \cite{Abdallah2021} analyze bundle-size pricing as the number of products grows. For both additive and nonadditive values, I establish conditions under which quantity-based menus maximize expected revenue among all feasible incentive-compatible and individually rational mechanisms. 

Related work also establishes conditions for optimal bundling. \cite{MANELLI20061} provide sufficient conditions for deterministic bundle-price schedules to maximize expected revenue among all incentive-compatible and individually rational mechanisms for a general number of objects, with applications to two and three objects. For pure bundling, \cite{MENICUCCI201533} and \cite{HaphpanahHartline2021} establish sufficient conditions, while \cite{Ghili2023} provides a characterization with nonadditive values under monotonic preferences and single-peaked profits. \cite{BergemannBonattiHauptSmolin2021} characterize optimal upgrade bundles and prices under ordering conditions on a finite type space. \cite{10.1257/aer.20230455} establishes conditions for optimal nested menus, while \cite{feng2025constructive} presents a constructive characterization of optimal bundling. \cite{Zhong2026} establishes optimality results under complementary preferences.  

The analysis also relates to rank-based mechanisms and multi-unit screening. \citet{BikhchandaniMishra2024} establish the revenue equivalence applied in my additive baseline and derive sufficient conditions for revenue monotonicity. \citet{BIKHCHANDANI2022} study two identical objects, while \citet{Pavlov2025selling} studies two units customizable into two versions, with the same private value for both units of each version. \citet{MalakhovVohra2009} and \cite{Devanuratal2020} study private per-unit values and quantity demand. My extensions allow distinct decreasing private marginal values across successive units of each product and derive optimal quantity prices under the corresponding distributional and support conditions. 

My verification result relates to duality-based methods \citep{DaskalakisDeckelbaumTzamos2017,KleinerManelli2019,cai2016duality}. I verify optimality by checking quantity virtual-surplus maximization and zero lower-boundary rent. It also relates to marginal-revenue methods \citep{Bulow1989,ASFHH2013}. Whereas \citet{Chen2025tale} analyze changes in bundle prices, my marginal comparisons concern virtual-surplus differences between consecutive quantities. 

\noindent\textbf{Paper organization.} Sections \ref{PremitivesAchaosIsALadder}--\ref{TheChangeWeNeed} introduce the model, assumptions, and rank-allocation reduction. Section \ref{GapFlowUpperBound} derives the revenue bound, the verification results, and the multiple-unit and state-dependent extensions. Section \ref{ExactOptimalityApplications} presents applications, and Section \ref{NonadditiveQuantityExtension} extends the analysis to bundle values. 

\end{spacing}

\section{Primitives and definitions}\label{PremitivesAchaosIsALadder}
\subsection{Additive environment, ranked values, and rank gaps}
In the additive baseline, a seller wishes to sell $K$ heterogeneous objects, indexed by $\mathcal K:=\{1,\ldots,K\}$, to a buyer. Both parties are risk neutral.  For every object $k\in\mathcal K$, the buyer's private value is $\theta_k\in[\underline\theta,\overline\theta]\subseteq\mathbb R_+$. Define the ambient type space by
$\overline\Theta:=[\underline\theta,\overline\theta]^K.$
The buyer's private values are represented as a type vector  $\theta:=(\theta_1,\dots,\theta_K)\in\Theta\subseteq\overline\Theta$. For every bundle $\widetilde{\mathcal K}\subseteq\mathcal K$, the buyer's value is $\sum_{k\in\widetilde{\mathcal K}}\theta_k$. All vectors are column vectors.

The type $\theta$ is drawn from a distribution $F$ with compact support $\Theta$.\footnote{Lemma \ref{UnboundedSingleActiveExactness} and Examples \ref{UnboundedLeadingExample} and \ref{ConditionalExponentialObjectValues} allow unbounded support on $\mathbb R_+^K$ while maintaining finite expected total value.} The support is permutation invariant, i.e., $\theta\in\Theta$ implies $\rho\theta\in\Theta$ for every $\rho\in\Sigma_{\mathcal K}$, where $\Sigma_{\mathcal K}$ is the set of all permutations of $\mathcal K$. The distribution admits a density $f$ with respect to $K$-dimensional Lebesgue measure on $\overline\Theta$. The density $f$ is positive almost everywhere on $\operatorname{int}\Theta$ and zero almost everywhere on $\overline\Theta\setminus\Theta$. The outside options of the seller and buyer are normalized to zero. Transfers from the buyer to the seller are nonnegative.

For each $\theta$, let $\theta^\downarrow:=(\theta_{(1)},\dots,\theta_{(K)})$ denote its nonincreasing rearrangement, with $\theta_{(r)}$ the $r$th-highest value. Ties occur with probability zero.\footnote{A fixed order over object labels resolves ties in the definition of $\Gamma(\theta)$.} Let $\Gamma(\theta)$ be the associated permutation matrix such that $\theta^\downarrow=\Gamma(\theta)\theta$. For any object-label allocation vector $q(\theta)\in[0,1]^K$, define the co-sorted rank allocation by $x(\theta):=\Gamma(\theta)q(\theta)$. Denote by $\langle\cdot,\cdot\rangle$  the Euclidean inner product.
Let $e_0:=\mathbf{0}\in\mathbb R^K$ and, for $n=1,\dots,K$, define $e_n$ as the vector with ones in its first $n$ entries and zeros elsewhere, i.e.,
$$e_n:=(\underbrace{1,\dots,1}_{n},\underbrace{0,\dots,0}_{K-n})^\top.$$

Next, I express ranked values in rank-gap coordinates. Let $\Theta^\downarrow:=\{\theta^\downarrow:\theta\in\Theta\}$ be the sorted support of $\theta^\downarrow$, and let $G$ and $g$ denote the distribution and density of $\theta^\downarrow$. For every $\theta^\downarrow\in\Theta^\downarrow$, define the rank gaps by
$$\delta_j:=\theta_{(j)}-\theta_{(j+1)},\;\;j=1,\dots,K,$$
with $\theta_{(K+1)}:=0$.\footnote{The terminal gap $\delta_K=\theta_{(K)}$ is measured relative to zero.} Hence, we have $\theta_{(r)}=\sum_{j=r}^K\delta_j$ for $r=1,\dots,K$.
Let $\Omega$ be the support of $\delta:=(\delta_1,\dots,\delta_K)$, and let $\gamma$ be the density of $\delta$. For every $j=1,\dots,K$, let $I_j:=\{\delta_j:\delta\in\Omega\}$ denote the coordinate projection of $\Omega$. 

Define the unit upper-triangular matrix
\begin{equation*}
\mathbf T:=
\begin{pmatrix}
1&1&\cdots&1\\
0&1&\cdots&1\\
\vdots&\vdots&\ddots&\vdots\\
0&0&\cdots&1
\end{pmatrix}.
\end{equation*}
The ranked and original value vectors satisfy $\theta^\downarrow=\mathbf T\delta$ and $\theta=\Gamma(\theta)^\top\mathbf T\delta$, with $\delta=\mathbf T^{-1}\theta^\downarrow$. Since $\det\mathbf T=1$, the transformation has unit Jacobian. Lemma \ref{ChangeWeNeed} in the Appendix shows the corresponding change-of-variables identity, so $\gamma(\delta)=g(\theta^\downarrow(\delta))$ for Lebesgue-almost every $\delta\in\Omega$.

\subsection{Quantity values and virtual surpluses}
For $n=0,\dots,K$, define the quantity value by
$$V_n(\delta):=\sum_{j=1}^K\min\{n,j\}\delta_j.$$
Conditional on the other gaps, a feasible increase of $\delta_j$ by $t$ raises the first $j$ ranked values by $t$ and increases $V_n$ by $\min\{n,j\}t$. The first differences are ranked values, $\theta_{(n)}=V_n(\delta)-V_{n-1}(\delta)$ for $n=1,\dots,K$. The second differences are rank gaps: $\delta_j=2V_j(\delta)-V_{j-1}(\delta)-V_{j+1}(\delta)$ for $j<K$, with $\delta_K=V_K(\delta)-V_{K-1}(\delta)$. Let $\bm V:=(V_n(\delta))_{n=1}^K$ be the column vector of quantity values. Then
\begin{equation*}
\begin{aligned}
\bm V&=\mathbf T^\top\theta^\downarrow
=\mathbf T^\top\mathbf T\delta,\\
\delta&=(\mathbf T^\top\mathbf T)^{-1}
\mathbf T^\top\theta^\downarrow
=(\mathbf T^\top\mathbf T)^{-1}\bm V.
\end{aligned}
\end{equation*}

For each gap $j$, let $Z_j$ denote the remaining gap coordinates. All the almost-everywhere statements about $Z_j$ are with respect to its marginal distribution. When the conditional section is a single nonempty interval, write
\begin{equation}\label{ConditionalGapSection}
\{z:(z,Z_j)\in\Omega\}=[\underline\delta_j(Z_j),\overline\delta_j(Z_j)]
\end{equation}
and define its conditional inverse-hazard ratio by
\begin{equation}\label{ConditionalGapInverseHazard}
R_j(z,Z_j):=\frac{\int_z^{\overline\delta_j(Z_j)}\gamma(s,Z_j)ds}{\gamma(z,Z_j)}.
\end{equation}
This is analogous to the single-object inverse-hazard ratio $(1-F)/f$ in \cite{Myerson81}.
When the conditional inverse-hazard ratio is state independent, write $R_j(z,Z_j)=R_j(z)$. 

For a given one-shot mechanism, information rent is the buyer's expected net utility conditional on his type. To relate this rent to allocations along different rank-gap directions, I introduce an auxiliary weight vector $a$ with $a_j$ corresponding to gap $j$. Denote the set of weight vectors by $\Delta_K:=\{a\in\mathbb R_+^K:\sum_{j=1}^Ka_j=1\}$.\footnote{Section~\ref{GapFlowUpperBound} explains why these weights are nonnegative and sum to one.} For every $a\in\Delta_K$, define the rank-gap virtual surplus by $$\phi_j^a(z):=z-a_jR_j(z),\;\; j=1,\dots,K.$$ 
For every rank $r$, define the rank virtual surplus by
\begin{equation}\label{RankVirtualSurplus}
\Phi_{(r)}^a(\delta):=\sum_{j=r}^K\phi_j^a(\delta_j),\;\;r=1,\dots,K.
\end{equation}
For every quantity $n$, define the cumulative quantity virtual surplus by
\begin{equation}\label{QuantityVirtualSurplus}
\Psi_n^a(\delta):=\sum_{r=1}^n\Phi_{(r)}^a(\delta)=\sum_{j=1}^K\min\{n,j\}\phi_j^a(\delta_j),\;\;n=0,\dots,K,
\end{equation}
with $\Psi_0^a(\delta)=0$. For $n=1,\dots,K$,  the rank virtual surplus is the change in quantity virtual surplus when quantity increases from $n-1$ to $n$, i.e., $\Phi_{(n)}^a(\delta)=\Psi_n^a(\delta)-\Psi_{n-1}^a(\delta)$. 
Equivalently, quantity virtual surplus is the value of the $n$ highest-valued objects minus a weighted rent: 
$$\Psi_n^a(\delta)=\sum_{r=1}^n\theta_{(r)}-\sum_{j=1}^K\min\{n,j\}a_jR_j(\delta_j),$$
where the coefficient $\min\{n,j\}$ counts the allocated ranks affected by increasing gap $j$. 
Likewise, let $\bm\phi^a:=(\phi_j^a(\delta_j))_{j=1}^K$ and $\bm\Psi^a:=(\Psi_n^a(\delta))_{n=1}^K$ be the corresponding column vectors. Then we have
\begin{equation*}
\bm\Psi^a=\mathbf T^\top\mathbf T\bm\phi^a,
\;\;
\bm\phi^a=(\mathbf T^\top\mathbf T)^{-1}\bm\Psi^a.
\end{equation*}

Next, fix $a\in\Delta_K$. For each gap $j$ with $a_j>0$, a critical cutoff (critical value) is a point $c_j\in\operatorname{int}(I_j)$ satisfying
\begin{equation}\label{ZeroCrossingVS}
c_j-a_jR_j(c_j)=0.
\end{equation}
Set $c_j=0$ whenever $a_j=0$, and write $c=(c_1,\dots,c_K)$ for any resulting critical vector. 
Let $\mathcal C_\gamma$ be the set of all critical vectors obtained as $a$ ranges over $\Delta_K$. For each $c\in\mathcal C_\gamma$, denote the corresponding weight vector by $a^c$. This weight vector is unique, and \eqref{ZeroCrossingVS} implies
\begin{equation*}
a_j^c=
\begin{cases}
\dfrac{c_j}{R_j(c_j)},&c_j>0,\\
0,&c_j=0.
\end{cases}
\end{equation*}
 We refer to  $V_n(c)=\sum_{j=1}^K\min\{n,j\}c_j$ as the cumulative critical value for quantity $n$.

For each quantity $n\in\{0,\dots,K\}$, define its critical comparison cell by 
\begin{equation*}
\mathscr C_n(c)
:=
\left\{
\delta\in\Omega:
\sum_{j=1}^K
\big(\min\{n,j\}-\min\{m,j\}\big)(\delta_j-c_j)
\geq0
\;\;\forall m=0,\dots,K
\right\}.
\end{equation*}
The cell $\mathscr C_n(c)$ contains the types for which $n$ maximizes $V_n(\delta)-V_n(c)$ over all quantities.
Define the set of support-relevant quantities, whose critical comparison cells have positive probability, by 
\begin{equation*}
\mathcal N(c)
:=
\left\{
n\in\{0,\dots,K\}:
\int_{\mathscr C_n(c)}\gamma(\delta)d\delta>0
\right\}.
\end{equation*}
The active gap set $\{j:a_j^c>0\}$ identifies the rent identities included in the weighted average. Only these directions contribute lower-boundary  rents. The zero-weighted  gaps remain part of the buyer's private information and continue to enter the quantity values.  

\section{Assumptions and distribution families}\label{AssumptionsAndFamilies}

\subsection{Baseline assumptions}

\begin{assumption}[Exchangeability]\label{ex}
The density $f$, extended by zero outside $\Theta$, is exchangeable, i.e., $f(\theta_1,\dots,\theta_K)=f(\theta_{\rho(1)},\dots,\theta_{\rho(K)})$ for Lebesgue-almost every $\theta\in\overline\Theta$ and every $\rho\in\Sigma_{\mathcal K}$.
\end{assumption}
Assumption \ref{ex} includes the cases of i.i.d. values and allows correlation through a common  factor that affects all objects symmetrically \citep[see, e.g.,][]{BikhchandaniMishra2024}.

Assumption \ref{ex}  is imposed   in the additive heterogeneous-object analysis and implies $g(\theta^\downarrow)=K!f(\theta^\downarrow)$ for almost every strictly ordered ranked type. The bundle-value formulation in Section \ref{NonadditiveQuantityExtension} and the direct homogeneous-unit formulation in Remark \ref{HomogeneousUnitsWithoutExchangeability} do not require this ex ante symmetry. Under Assumption \ref{ex}, the conditional inverse-hazard ratio at the gap level in \eqref{ConditionalGapInverseHazard} can also be expressed in the original object-value coordinates as follows. For almost every $\theta\in\Theta$ with $f(\theta)>0$ and induced gaps $\delta$, 
\begin{equation}\label{RinTheta}
R_j(\delta_j,Z_j)=\frac{1}{f(\theta^\downarrow)}\int_0^{\overline\delta_j(Z_j)-\delta_j}f(\theta^\downarrow+t e_j)\,dt.
\end{equation}
The shift $t e_j$ increases the $j$ highest object values by $t$, leaving the others unchanged. 

\begin{assumption}[Gap rent regularity]\label{MightAsWelllWalkingOnTheSun}
For every gap $j$ and almost every $Z_j$, the conditional section is the single interval in \eqref{ConditionalGapSection}. On the relative interior of every nondegenerate section, the conditional inverse-hazard ratio is well defined and satisfies $R_j(z,Z_j)=R_j(z)$ for every relevant $z$. Moreover, $R_j$ admits a continuous extension to $I_j$.
\end{assumption}

The following conditions on the gap density and its support imply Assumption \ref{MightAsWelllWalkingOnTheSun}.

{
\renewcommand{\theassumption}{2$'$}
\renewcommand{\theHassumption}{2prime}
\begin{assumption}[Endpoint-anchored product density]\label{FlyMeToTheMoon}
\textup{(i)} The support $\Omega$ is nonempty and compact, $\operatorname{int}\Omega$ is path connected, and $\Omega=\operatorname{cl}(\operatorname{int}\Omega)$. For each $j$, $I_j=[\ell_j,\overline\delta_j]$ with $\ell_j<\overline\delta_j$. For almost every $Z_j$, the conditional section has the form \eqref{ConditionalGapSection}, with upper endpoint $\overline\delta_j$ and lower endpoint $\underline\delta_j(Z_j)\geq\ell_j$.

\textup{(ii)} The gap density satisfies
\begin{equation}\label{PRoductDensitty}
\gamma(\delta)=C\mathbf1_\Omega(\delta)\prod_{j=1}^K\psi_j(\delta_j),
\end{equation}
where $C:=\left(\int_\Omega\prod_{j=1}^K\psi_j(s_j)\,ds\right)^{-1}>0$ is the normalizing constant. Each $\psi_j$ is positive, continuous, and integrable on $[\ell_j,\overline\delta_j)$, and its upper mass $S_j(z):=\int_z^{\overline\delta_j}\psi_j(s)ds$ is positive and log-concave on that interval.

\textup{(iii)} The lower endpoints satisfy $\sum_{j=1}^K\ell_j\psi_j(\ell_j)/S_j(\ell_j)<1$.
\end{assumption}
}
\addtocounter{assumption}{-1}
The product form and fixed conditional upper endpoints imply $R_j(z,Z_j)=S_j(z)/\psi_j(z)$, so each conditional inverse-hazard ratio depends only on the corresponding gap. State-dependent lower endpoints may nevertheless leave the gaps statistically dependent. Lemma \ref{GapFactorizationSeparability} establishes that each $R_j$ is continuous, nonincreasing, and becomes zero at its upper endpoint. Write $h_j(z):=\psi_j(z)/S_j(z)$ for the corresponding hazard rate. Lemma \ref{PrimitiveCriticalVectorexistence} establishes the existence of full-active and single-active critical vectors.

Although the standard i.i.d. object-value distributions are exchangeable, they need not satisfy gap rent regularity in Assumption \ref{MightAsWelllWalkingOnTheSun}, as shown by the following remark. 

 \begin{remark}\label{IIDStateDependentRemark}
Let $K\geq2$ and suppose the object values are i.i.d. with distribution $F_0$ and positive density $f_0$ on $[\underline\theta,\overline\theta]$, where $\overline\theta<\infty$. Conditional on $Z_1$, $\theta_{(2)}=\sum_{j=2}^K\delta_j$ and
$$R_1(z,Z_1)=\frac{1-F_0(z+\theta_{(2)})}{f_0(z+\theta_{(2)})}.$$
State independence would require a constant marginal inverse hazard, which is incompatible with a finite upper endpoint. Thus Assumption \ref{MightAsWelllWalkingOnTheSun} fails for bounded i.i.d. object values. If compactness is relaxed, the i.i.d. exponential object values studied by \citet{Giannakopoulos2015} satisfy state independence but violate Assumption \ref{YouDontDeserveIt}  even with a single active direction. Lemma \ref{UnboundedSingleActiveExactness} instead keeps the inactive gaps bounded and does not impose i.i.d. object values.
\end{remark}

\subsection{Sufficient conditions for optimal quantity pricing}

The following assumptions can be verified directly from the primitive  distribution and its support for a specified critical vector $c\in\mathcal C_\gamma$.\footnote{For nonadditive bundle values, \citet[Theorem~2]{Zhong2026} also establishes sufficient conditions that can be verified from the primitive distribution and its support for a given candidate menu.}$^,$\footnote{The verification conditions in Section~\ref{GapFlowUpperBound} instead concern quantity prices and the induced menu.} When imposed jointly, the assumptions here refer to the same $c$; they need not hold for every critical vector. 

\begin{assumption}[Critical-partition virtual-surplus alignment]\label{TopRankSignAlignment}
For the specified $c\in\mathcal C_\gamma$, assume  $\Psi_n^{a^c}(\delta)\geq\Psi_m^{a^c}(\delta)$ for every $n\in\mathcal N(c)$, every $m\in\{0,\dots,K\}$, and almost every $\delta\in\mathscr C_n(c)$.
\end{assumption}
Assumption \ref{TopRankSignAlignment} requires every support-relevant maximizer of $V_n(\delta)-V_n(c)$ to maximize quantity virtual surplus almost surely. 
The following identity compares the centered cumulative values with the corresponding quantity virtual surpluses. For any $c\in\mathcal C_\gamma$, any $n,m$, and almost every $\delta\in\Omega$,
\begin{equation}\label{CriticalComparisonIdentity}
\begin{split}
\Psi_n^{a^c}(\delta)-\Psi_m^{a^c}(\delta)
={}&
\big(V_n(\delta)-V_n(c)\big)
-\big(V_m(\delta)-V_m(c)\big)\\
&+\sum_{\substack{j=1\\a_j^c>0}}^K
\big(\min\{n,j\}-\min\{m,j\}\big)
a_j^c\big(R_j(c_j)-R_j(\delta_j)\big).
\end{split}
\end{equation}

{
\renewcommand{\theassumption}{3$'$}
\renewcommand{\theHassumption}{3prime}
\begin{assumption}[Active local inverse-hazard alignment]\label{ActiveLocalInverseHazardAlignment}
Suppose Assumption \ref{MightAsWelllWalkingOnTheSun} holds. For the specified $c\in\mathcal C_\gamma$, the following conditions hold. 

\textup{(i)} If $0\in\mathcal N(c)$, then $\Psi_m^{a^c}(\delta)\leq0$ for every $m=1,\dots,K$ and almost every $\delta\in\mathscr C_0(c)$.

\textup{(ii)} For every $n\in\mathcal N(c)\setminus\{0\}$, every $j\leq n$ with $a_j^c>0$, and almost every $\delta\in\mathscr C_n(c)$, we have $R_j(\delta_j)\leq R_j(c_j)$.

\textup{(iii)} For every $n\in\mathcal N(c)\setminus\{0\}$, every $j>n$ with $a_j^c>0$, and almost every $\delta\in\mathscr C_n(c)$, we have $R_j(\delta_j)=R_j(c_j)$.
\end{assumption}
}
\addtocounter{assumption}{-1}

 Part \textup{(i)} ensures that quantity zero maximizes cumulative quantity virtual surplus on $\mathscr C_0(c)$ almost surely. On each positive-quantity cell $\mathscr C_n(c)$, parts \textup{(ii)} and \textup{(iii)}  ensure that the last sum in \eqref{CriticalComparisonIdentity} is nonnegative, and thus quantity $n$ also maximizes cumulative quantity virtual surplus $\Psi_n^{a^c}(\delta)$ almost surely.  That is, Assumption \ref{ActiveLocalInverseHazardAlignment} implies Assumption \ref{TopRankSignAlignment} (see Lemma \ref{ActiveLocalInverseHazardAlignmentImpliesAlignment} in the Appendix).

\begin{assumption}[Active lower-face support]\label{YouDontDeserveIt}
For the same $c\in\mathcal C_\gamma$ satisfying Assumption \ref{TopRankSignAlignment}, assume that $V_n(\underline\delta_j(Z_j),Z_j)-V_n(c)\leq0$ for every $j$ with $a_j^c>0$, almost every $Z_j$, and every $n\in\mathcal N(c)\setminus\{0\}$.
\end{assumption}

Assumption \ref{YouDontDeserveIt} requires the buyer to receive no information rent at the smallest possible value of each active gap, conditional on other gaps.  In the original object-value coordinates, this condition is equivalent to
$$
\sum_{r=1}^n\theta_{(r)}-\min\{n,j\}\big(\delta_j-\underline\delta_j(Z_j)\big)\leq V_n(c)
$$
for every active $j$, every $n\in\mathcal N(c)\setminus\{0\}$, and $F$-almost every $\theta$ with induced gaps $\delta$. Assumptions \ref{TopRankSignAlignment} and \ref{YouDontDeserveIt} provide sufficient conditions for the  exact optimality result under additive values. My main result, Theorem \ref{EndogenousGapFlowGap}, also allows other price vectors, as illustrated, for example, by the optimal menu for Distribution Family \ref{GaussianGapDensityFamily} and the two-part-tariff extension in Proposition \ref{LocalMarginalMatchingExactness} in the Appendix.

\subsection{Distribution families} 

The following  distribution families satisfy Assumptions \ref{ex} and \ref{MightAsWelllWalkingOnTheSun}. For each family, let $\gamma$ have the product-density form \eqref{PRoductDensitty} with the specified support and factors, and let the object-value density be the exchangeable symmetrization of the corresponding ranked-value density. Distribution Family \ref{GaussianGapDensityFamily} violates Assumption \ref{YouDontDeserveIt} but satisfies Conditions \ref{FiniteVirtualSurplusImplementation} and \ref{ZeroBoundaryResidual} of Theorem \ref{EndogenousGapFlowGap}. Distribution Families \ref{SingleActiveGapDensityFamily}, \ref{FullActiveGapDensityFamily}, and \ref{TwoActiveGapDensityFamily} satisfy Assumptions \ref{FlyMeToTheMoon}, \ref{TopRankSignAlignment} (or \ref{ActiveLocalInverseHazardAlignment}), and \ref{YouDontDeserveIt}. For each bounded family, choose $\underline\theta=0$ and $\overline\theta$ large enough to contain the induced object-value support. 

\begin{distributionfamily}[$K\geq2$, joint truncation of nonincreasing log-concave factors]\label{GaussianGapDensityFamily}
Let $K\geq2$, set $a_K=1$ and $a_j=0$ for $j<K$, and fix $0<\overline\delta_j<\infty$ for every $j$. Let $\psi_j$ be a positive, continuous, nonincreasing, log-concave density factor on $[0,\overline\delta_j]$. Let $\underline\theta=0$ and $\overline\theta\geq\sum_{j=1}^K\overline\delta_j$. Let $\upsilon\in(0,\overline\delta_K)$ be the unique solution of $\upsilon\psi_K(\upsilon)=S_K(\upsilon)$. Choose $\varpi$ such that $\max\{0,\overline\delta_1-(K-1)\upsilon\}\leq\varpi<\overline\delta_1$, and define $\ell_K=\upsilon+(\varpi-\overline\delta_1)/(K-1)$. Require
\begin{equation}\label{GaussianLastGapBound}
\overline\delta_{K-1}+\ell_K\geq R_K(\ell_K).
\end{equation}
The gap support is
\begin{equation}\label{GaussianGapSupport}
\Omega=\left\{\delta\in\prod_{j=1}^K[0,\overline\delta_j]:
\max_{n\in\{K-1,K\}}
\left\{V_n(\delta)-\sum_{j=2}^{K-1}j\overline\delta_j-\varpi-n\upsilon\right\}
\geq0\right\}.
\end{equation}
The terminal-coordinate projection is $I_K=[\ell_K,\overline\delta_K]$. Assumption \ref{MightAsWelllWalkingOnTheSun} holds, with $R_j(z)=S_j(z)/\psi_j(z)$. Apart from \eqref{GaussianLastGapBound}, the intermediate-coordinate upper bounds are unrestricted.

Constant factors yield jointly uniform object values on the support induced by \eqref{GaussianGapSupport}. Gaussian factors $\psi_j(z)=\exp[-z^2/(2\varsigma_j^2)]$, with $\varsigma_j>0$, yield a diagonal-Gaussian density in gap coordinates, jointly truncated to $\Omega$. This truncation generally induces dependence. Log-concavity is imposed on the individual factors, not on the jointly truncated density. See Subsection~\ref{ApplicationJointTruncation} for the optimal results. Example~\ref{JointUniformLeadingExample} shows a two-object uniform specification.
\end{distributionfamily}

\begin{distributionfamily}[$K\geq2$, an exponential initial segment]\label{SingleActiveGapDensityFamily}
Let $K\geq2$, fix $\iota\in\{1,\dots,K\}$, and let $\Omega=\prod_{j=1}^K[0,\overline\delta_j]$, where $0<\overline\delta_j<\infty$ for every $j$. For each $j$, let $\psi_j$ be a positive, continuous, integrable density factor on $[0,\overline\delta_j)$ with log-concave upper mass $S_j$. Define $\tau:=S_{\iota}(0)/\psi_{\iota}(0)$ and suppose that
\begin{equation}\label{SingleActiveDensitySupportBounds}
0<\tau<\overline\delta_{\iota},\;\; \sum_{j\neq\iota}\overline\delta_j\leq\tau,\;\; \sum_{j\neq\iota}j\overline\delta_j\leq\iota\tau.
\end{equation}
Assume that the active factor has the exponential initial segment, i.e.,
\begin{equation}\label{SingleActiveDensityInitialSegment}
\psi_{\iota}(z)=\psi_{\iota}(0)e^{-z/\tau},\;\; 0\leq z\leq\tau.
\end{equation}
For this bounded support, Assumptions \ref{FlyMeToTheMoon}, \ref{ActiveLocalInverseHazardAlignment}, and \ref{YouDontDeserveIt} hold with
\begin{equation}\label{SingleActiveCriticalCoordinates}
c_{\iota}=\tau,\;\; c_j=0\;(j\neq\iota),\;\; a_{\iota}^c=1,\;\; \text{and}\;\; a_j^c=0\;(j\neq\iota).
\end{equation}
The distribution remains full dimensional.\footnote{Example~\ref{UnboundedLeadingExample} illustrates the unbounded extension in Lemma~\ref{UnboundedSingleActiveExactness}; Example~\ref{ConditionalExponentialObjectValues} shows an exponential case.} See Subsection~\ref{ApplicationSingleActive} for the optimality results. 
\end{distributionfamily}

\begin{distributionfamily}[$K\geq2$, independent reverse-Beta gaps]\label{FullActiveGapDensityFamily}
Let $K\geq2$, fix $\overline\delta>0$, and let $\Omega=[0,\overline\delta]^K$. For every $j=1,\dots,K$, assume that $\psi_j(z)=\frac{\alpha}{\overline\delta^\alpha}(\overline\delta-z)^{\alpha-1},\;\;0\leq z<\overline\delta,$ where $0<\alpha\leq\frac{1}{K\left(\frac{K(K+1)}{2}-1\right)}.$ Then Assumptions \ref{FlyMeToTheMoon}, \ref{TopRankSignAlignment}, and \ref{YouDontDeserveIt} hold with $c_j=\frac{\overline\delta}{1+K\alpha},\;\;a_j^c=\frac1K,\;\;j=1,\dots,K.$ See Subsection~\ref{ApplicationQuantityDiscounts} for the optimal pricing results for this family.
\end{distributionfamily}

\par\smallskip\noindent\begin{minipage}{\linewidth}
\setstretch{1.16}
\setlength{\abovedisplayskip}{6pt}
\setlength{\belowdisplayskip}{6pt}
\setlength{\abovedisplayshortskip}{6pt}
\setlength{\belowdisplayshortskip}{6pt}

\begin{example}[Reverse-Beta gaps in three coordinate systems]\label{ReverseBetaThreeSpaces}
Let $K=2$ and consider independent gaps $\delta_1,\delta_2\sim\mathrm{Beta}(1,1/4)$. Their joint density is
$$
\gamma(\delta)=\frac1{16}\big[(1-\delta_1)(1-\delta_2)\big]^{-3/4},\;\; 0\leq\delta_1,\delta_2<1.
$$
Since $\theta_{(1)}=\delta_1+\delta_2$ and $\theta_{(2)}=\delta_2$, the ranked-value density is
$$
g(\theta^\downarrow)=\frac1{16}\big[(1-\theta_{(1)}+\theta_{(2)})(1-\theta_{(2)})\big]^{-3/4},\;\; 0\leq\theta_{(2)}<1,\; \theta_{(2)}\leq\theta_{(1)}<\theta_{(2)}+1.
$$
By Assumption \ref{ex}, the original object-value density is
$$
f(\theta)=\frac1{32}\big[(1-|\theta_1-\theta_2|)(1-\min\{\theta_1,\theta_2\})\big]^{-3/4},\;\; \theta\in\operatorname{int}\Theta,
$$
where $\Theta=\{\theta\in\mathbb R_+^2:|\theta_1-\theta_2|\leq1,\ \min\{\theta_1,\theta_2\}\leq1\}$. These density formulas hold almost everywhere and vanish outside their respective supports. Figure~\ref{ReverseBetaThreeSpacesFigure} illustrates the three supports; the colors show the optimal quantity choices derived in Subsection~\ref{ApplicationQuantityDiscounts}.
\end{example}
\begin{figure}[H]
\setstretch{1}
\centering
\includegraphics[width=\textwidth]{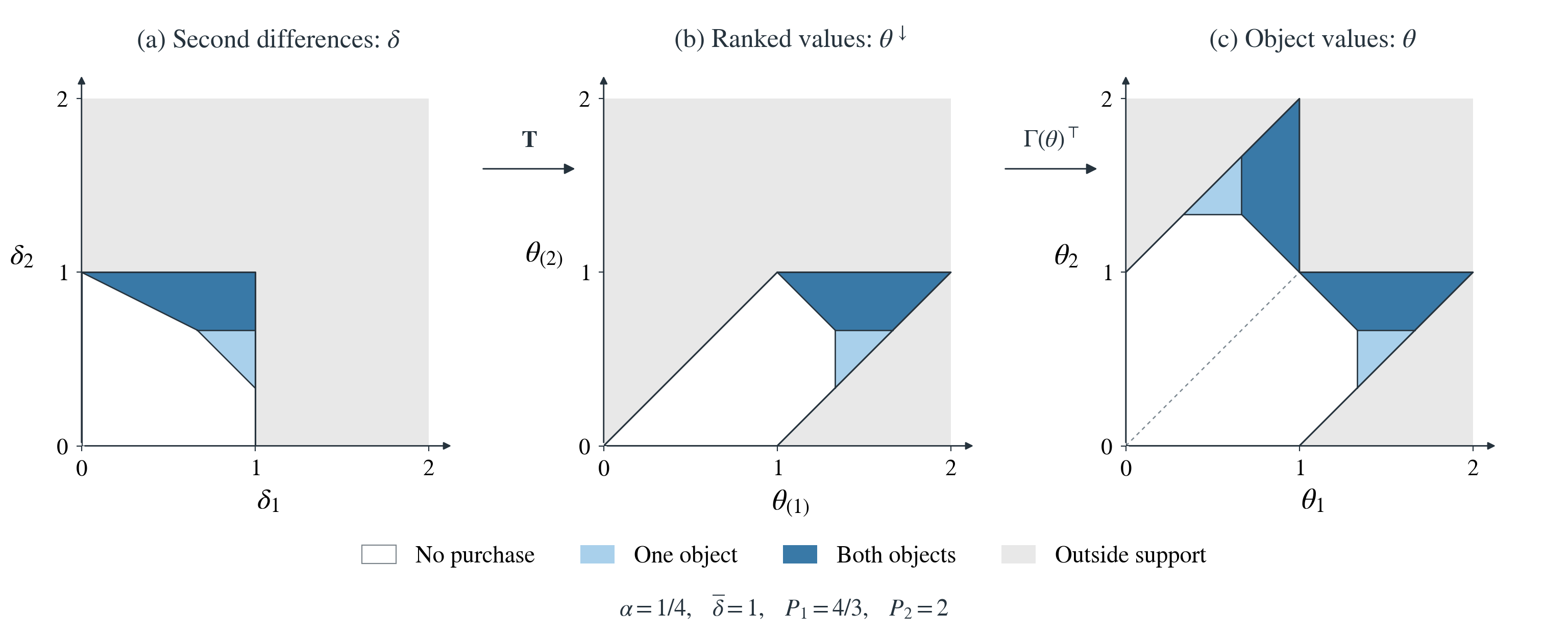}
\caption{The reverse-Beta example in second-difference, ranked-value, and object-value coordinates. White, light blue, and dark blue denote quantities $0$, $1$, and $2$; gray denotes points outside the support. The right panel shows both label assignments.}
\label{ReverseBetaThreeSpacesFigure}
\end{figure}
\end{minipage}\par\smallskip

\begin{distributionfamily}[$K\geq2$, log-concave upper masses with exponential initial segments in the first and terminal factors]\label{TwoActiveGapDensityFamily}
Let $K\geq2$. Fix $0\leq\ell_j<\overline\delta_j<\infty$ for $j=1,\dots,K$, with $\ell_j=0$ for $j=2,\dots,K-1$, such that
\begin{equation}\label{TwoActiveDensitySupportBounds}
\overline\delta_K-\ell_K+\sum_{j=2}^{K-1}(j-1)\overline\delta_j<\overline\delta_1-\ell_1<K(\overline\delta_K-\ell_K)-\sum_{j=2}^{K-1}(K-j)\overline\delta_j.
\end{equation}
Let $\beta_1:=\ell_K+\sum_{j=1}^{K-1}\overline\delta_j$ and $\beta_K:=\ell_1+\sum_{j=2}^Kj\overline\delta_j$. Let the gap support be
\begin{equation}\label{TwoActivePrimitiveSupport}
\Omega=\bigcup_{n=1}^K\left\{\delta\in\prod_{j=1}^K[\ell_j,\overline\delta_j]:V_n(\delta)\geq\frac{(K-n)\beta_1+(n-1)\beta_K}{K-1}\right\}.
\end{equation}

Fix $a\in\Delta_K$ with $a_1,a_K>0$ and $a_j=0$ for $j=2,\dots,K-1$. For every $j$, let $\psi_j$ be a positive, continuous, integrable density factor on $[\ell_j,\overline\delta_j)$ with log-concave upper mass $S_j$. Assume that
$$\begin{aligned}
a_1\frac{S_1(\ell_1)}{\psi_1(\ell_1)}+a_K\frac{S_K(\ell_K)}{\psi_K(\ell_K)}&=\beta_1,\\
a_1\frac{S_1(\ell_1)}{\psi_1(\ell_1)}+Ka_K\frac{S_K(\ell_K)}{\psi_K(\ell_K)}&=\beta_K,
\end{aligned}$$
and, for $i\in\{1,K\}$,
\begin{equation}\label{TwoActiveDensityInitialSegments}
\psi_i(z)=\psi_i(\ell_i)\exp\left\{-\frac{\psi_i(\ell_i)}{S_i(\ell_i)}(z-\ell_i)\right\},\;\;\ell_i\leq z\leq a_i\frac{S_i(\ell_i)}{\psi_i(\ell_i)}.
\end{equation}

Then Assumptions \ref{FlyMeToTheMoon}, \ref{ActiveLocalInverseHazardAlignment}, and \ref{YouDontDeserveIt} hold with $a^c=a$ and
\begin{equation}\label{TwoActiveCriticalCoordinates}
c_i=a_i\frac{S_i(\ell_i)}{\psi_i(\ell_i)},\;\;i\in\{1,K\},\;\; c_j=0,\;\;j=2,\dots,K-1.
\end{equation}
The two positive coordinates are the unique zero crossings of the corresponding rank-gap virtual surpluses. The distribution remains full dimensional. See Subsection~\ref{ApplicationTwoPartTariff} for the optimal pricing results for this family.
\end{distributionfamily}

\section{Mechanisms and the rank-allocation reduction}\label{TheChangeWeNeed}

Throughout this study, mechanism rules and the allocation-payment pairs selected from menus are assumed to be Borel measurable.

\noindent\textbf{General \textit{direct} mechanisms.}   By the revelation  principle \citep{Myerson81}, it is without loss of generality to focus on \textit{direct} mechanisms where the buyer's action space  coincides with his type space. 

Define the feasible allocation set by $\mathcal Q:=[0,1]^K$. A general \textit{direct} mechanism is a  pair $(q,p):\Theta\to\mathcal Q\times\mathbb R_+$, where $q(\theta)=(q_1(\theta),\dots,q_K(\theta))$ specifies the allocation probabilities of the $K$ objects and $p(\theta)$ specifies the buyer's payment to the seller.

A general \textit{direct} mechanism $(q,p)$ is \textit{incentive compatible} (IC) if
\begin{equation*}
\langle q(\theta),\theta\rangle-p(\theta)\geq\langle q(\hat\theta),\theta\rangle-p(\hat\theta),\;\;\forall\theta,\hat\theta\in\Theta.\tag{IC}
\end{equation*}
A general \textit{direct} mechanism $(q,p)$ is \textit{individually rational} (IR) if
\begin{equation*}
\langle q(\theta),\theta\rangle-p(\theta)\geq0,\;\;\forall\theta\in\Theta.\tag{IR}
\end{equation*}

Let $\mathcal M_{\mathcal Q}$ be the set of all  general \textit{direct} mechanisms $(q,p)$ satisfying IC and IR with $q(\theta)\in\mathcal Q$ for every $\theta$. The seller's revenue from $(q,p)$ is defined by $\mathrm{REV}(q,p):=\mathbb E_\theta[p(\theta)]=\int_\Theta p(\theta)\,dF(\theta)$. Further define the optimal revenue over general \textit{direct} mechanisms by $\mathrm{REV}_{\mathcal Q}^{*}:=\sup_{(q,p)\in\mathcal M_{\mathcal Q}}\mathrm{REV}(q,p).$

\noindent\textbf{\textit{Direct} rank-allocation mechanisms.}
Denote the feasible monotone  rank-allocation set by
\begin{equation*}
\mathcal X:=\left\{x\in[0,1]^K\mid 1\geq x_{(1)}\geq x_{(2)}\geq\cdots\geq x_{(K)}\geq0\right\}.
\end{equation*}

A \textit{direct} rank-allocation mechanism is a  pair $(x,p):\Theta^\downarrow\to\mathcal X\times\mathbb R_+$. Given a labeled report $\theta\in\Theta$, the mechanism acts on the sorted report $\theta^\downarrow$ and induces the object-label allocation $q^x(\theta):=\Gamma(\theta)^\top x(\theta^\downarrow)$ and payment $p^x(\theta):=p(\theta^\downarrow)$. 
The mechanism $(x,p)$ is IC and IR if the induced \textit{direct} mechanism $(q^x,p^x)$ is IC and IR. Let $\mathcal M_{\mathcal X}$ be the set of feasible IC and IR \textit{direct} rank-allocation mechanisms, and define $\mathrm{REV}_{\mathcal X}^{*}:=\sup_{(x,p)\in\mathcal M_{\mathcal X}}\mathbb E[p(\theta^\downarrow)].$

The symmetrization and rank-preservation results of \cite[Proposition 1 and Theorem 2]{BikhchandaniMishra2024} establish the revenue equivalence between general mechanisms and rank-allocation mechanisms under exchangeability. The following lemma shows the corresponding result for the permutation-invariant support considered here. 

\begin{lemma}[Rank-allocation reduction]\label{QRankEquivalence}
Under Assumption \ref{ex}, we have $\mathrm{REV}_{\mathcal Q}^{*}=\mathrm{REV}_{\mathcal X}^{*}.$
\end{lemma}

\noindent\textbf{Top-rank quantity representation.} A quantity-based menu specifies prices by quantity rather than object labels, and the buyer selects the highest-valued objects to fill the chosen quantity. Lemma \ref{quantitycompress} in the Appendix formalizes this choice. The following lemma characterizes the monotone rank-allocation simplex and establishes its equivalence to report-independent top-rank menus.

If $x\in\mathcal X$ is represented by a top-rank quantity lottery $\chi\in\Delta(\{0,\dots,K\})$, the buyer's expected value from this lottery is $V_\chi(\theta):=\sum_{n=0}^{K}\chi_nV_n(\theta)$. 
A report-independent top-rank menu is $$M=\{(\chi^\nu,p^\nu):\nu\in\mathcal I\},$$ where $\chi^\nu\in\Delta(\{0,\dots,K\})$, $p^\nu\in\mathbb R_{+}$, and $\mathcal I$ is an index set.    Every type faces the same list of quantity lotteries and prices; the buyer's ranking determines which objects fill a realized quantity. An admissible menu is paired with a payoff-maximizing selection $\nu(\theta^\downarrow)\in\argmax_{\nu\in\mathcal I}\{V_{\chi^\nu}(\theta^\downarrow)-p^\nu\}$. Let $\mathcal M_{\mathrm{menu}}$ be the class of IC and IR self-selection mechanisms generated by such menus, and define  the corresponding optimal revenue by $\mathrm{REV}_{\mathrm{menu}}^{*}:=\sup_{M\in\mathcal M_{\mathrm{menu}}}\mathbb E[p_M(\theta)].$

\begin{lemma}[Top-rank simplex and menu representation]\label{RankMenuEquivalence}
The extreme points of $\mathcal X$ are $e_0,\dots,e_K$. Equivalently, every $x\in\mathcal X$ admits a unique lottery representation $x_{(r)}=\sum_{n=r}^K\chi_n$ for $r=1,\dots,K$, where $\chi_0=1-x_{(1)}$, $\chi_n=x_{(n)}-x_{(n+1)}$ for $n=1,\dots,K-1$, and $\chi_K=x_{(K)}$. Under Assumption \ref{ex},
$$
\mathrm{REV}_{\mathcal Q}^{*}=\mathrm{REV}_{\mathcal X}^{*}=\mathrm{REV}_{\mathrm{menu}}^{*}.
$$
Thus the seller can restrict attention to report-independent top-rank quantity lotteries. The results below establish conditions under which fixed quantities at posted prices are optimal. 
\end{lemma}

\section{Rank-gap decomposition and exact optimality}\label{GapFlowUpperBound}

By Lemma \ref{RankMenuEquivalence}, it suffices to study report-independent menus over top-rank quantity lotteries. Fix $M=\{(\chi^\nu,p^\nu):\nu\in\mathcal I\}\in\mathcal M_{\mathrm{menu}}$. I work in rank-gap coordinates and write $\chi(\delta)$, $x(\delta)$, and $p_M(\delta)$ for the selected quantity lottery, rank allocation, and payment. Define the cumulative rank allocation by
\begin{equation}\label{CumulativeRankAllocation}
Q_j(\delta):=\sum_{r=1}^j x_{(r)}(\delta)=\sum_{n=0}^K\min\{n,j\}\chi_n(\delta),\;\;j=1,\dots,K.
\end{equation}
Thus $Q_j(\delta)$ is the expected number of allocated objects among the first $j$ ranks. The buyer's indirect utility, or information rent, is expected value minus payment and has the equivalent representations
\begin{equation*}
\begin{split}
U(\delta)&:=\sup_{\nu\in\mathcal I}\left\{\sum_{n=0}^K\chi_n^\nu V_n(\delta)-p^\nu\right\}\\
&=\sum_{n=0}^K\chi_n(\delta)V_n(\delta)-p_M(\delta)\\
&=\sum_{r=1}^Kx_{(r)}(\delta)\theta_{(r)}(\delta)-p_M(\delta)\\
&=\sum_{j=1}^KQ_j(\delta)\delta_j-p_M(\delta).
\end{split}
\end{equation*}
Consequently,
\begin{equation}\label{RankGapExpectedRevenue}
\mathrm{REV}(M)=\mathbb E_\gamma\left[\sum_{j=1}^KQ_j(\delta)\delta_j-U(\delta)\right].
\end{equation}
Each menu payoff is affine in $\delta$, with uniformly bounded slopes. Since prices are nonnegative and the menu is nonempty, its supremum is finite at every finite gap vector. Hence $U$ admits a finite convex Lipschitz extension and is absolutely continuous on every bounded conditional interval. Individual rationality and nonnegative payments also yield $0\leq U(\delta)\leq V_K(\delta)$ and $0\leq p_M(\delta)\leq V_K(\delta)$ on $\Omega$. By Rademacher's theorem, this Lipschitz extension is differentiable Lebesgue-almost everywhere, and hence $\gamma$-almost surely. At such a type, the selected affine payoff touches $U$ from below, so their gradients agree. Consequently,
\begin{equation}\label{Envelope}
\partial_jU(\delta)=Q_j(\delta),\;\;j=1,\dots,K.
\end{equation}
\noindent\textbf{Conditional integration by parts.} Fix a gap $j$. For the following calculation, suppose that the conditional sections have the single-interval form in \eqref{ConditionalGapSection}, the ratios in \eqref{ConditionalGapInverseHazard} are well defined, and the associated terms are integrable. I first retain their full dependence on $Z_j$. For almost every conditioning state with positive finite section mass,
$$\mathbb E_\gamma[U(\delta)\mid Z_j]=\frac{\int_{\underline\delta_j(Z_j)}^{\overline\delta_j(Z_j)}U(z,Z_j)\gamma(z,Z_j)dz}{\int_{\underline\delta_j(Z_j)}^{\overline\delta_j(Z_j)}\gamma(z,Z_j)dz}.$$
By Fubini's theorem, \eqref{Envelope} holds almost everywhere along almost every conditional section. Since $\gamma(z,Z_j)dz=-d_z[\int_z^{\overline\delta_j(Z_j)}\gamma(s,Z_j)ds]$, absolute continuity of $U$ permits integration by parts \citep[see][]{Milgrom2002}:
\begin{equation}\label{ConditionalGapIBP}
\begin{split}
&\int_{\underline\delta_j(Z_j)}^{\overline\delta_j(Z_j)}U(z,Z_j)\gamma(z,Z_j)dz\\
=&-\left[U(z,Z_j)\int_z^{\overline\delta_j(Z_j)}\gamma(s,Z_j)ds\right]_{\underline\delta_j(Z_j)}^{\overline\delta_j(Z_j)}
+\int_{\underline\delta_j(Z_j)}^{\overline\delta_j(Z_j)}\partial_jU(z,Z_j)\left[\int_z^{\overline\delta_j(Z_j)}\gamma(s,Z_j)ds\right]dz.
\end{split}
\end{equation}
At a finite upper endpoint, the upper-tail integral becomes zero. The lower endpoint contributes utility times the section mass. Dividing \eqref{ConditionalGapIBP} by this mass and applying \eqref{ConditionalGapInverseHazard} and \eqref{Envelope} yields 
\begin{equation}\label{StateDependentConditionalRentIdentity}
\mathbb E_\gamma[U(\delta)\mid Z_j]=U(\underline\delta_j(Z_j),Z_j)+\mathbb E_\gamma[Q_j(\delta)R_j(\delta_j,Z_j)\mid Z_j].
\end{equation}
This conditional rent identity does not require state independence in Assumption \ref{MightAsWelllWalkingOnTheSun}. It separates utility at the conditional lower face from the rent term generated by the cumulative allocation along gap $j$.

Under Assumption \ref{MightAsWelllWalkingOnTheSun}, $R_j(\delta_j,Z_j)=R_j(\delta_j)$. By the extreme value theorem, $R_j$ is bounded on the compact projection $I_j$. Together with $0\leq Q_j\leq j$ and $0\leq U\leq V_K\leq K\overline\theta$, this establishes integrability of the rent and lower-boundary terms. Integrating over $Z_j$ in \eqref{StateDependentConditionalRentIdentity} and applying the law of iterated expectations yields 
\begin{equation}\label{ConditionalGapRentIdentity}
\mathbb E_\gamma[U(\delta)]=\mathbb E_\gamma[Q_j(\delta)R_j(\delta_j)]+\mathbb E_{Z_j}[U(\underline\delta_j(Z_j),Z_j)].
\end{equation}
For the fixed menu $M$, IC implies the envelope identity \eqref{Envelope}, and the integration above establishes \eqref{ConditionalGapRentIdentity} along each gap, always for the same expected information rent. Thus $\sum_{j=1}^Ka_j=1$ preserves $\mathbb E_\gamma[U]$ when the identities are averaged. IR ensures that each lower-boundary rent is nonnegative, and $a_j\geq0$ preserves this property in the weighted sum. These observations justify the restrictions on $a$ and yield the following revenue decomposition. 

\begin{proposition}[Gap revenue decomposition]\label{GapRevenueDecomposition}
Under Assumption \ref{MightAsWelllWalkingOnTheSun}, every $a\in\Delta_K$ and every IC and IR top-rank menu $M$ satisfy
\begin{equation}\label{GapRevenueIdentity}
\begin{split}
\mathrm{REV}(M)&=\mathbb E_\gamma\left[\sum_{j=1}^KQ_j(\delta)\phi_j^a(\delta_j)\right]-\mathcal B^a(U)\\
&=\mathbb E_\gamma\left[\sum_{r=1}^Kx_{(r)}(\delta)\Phi_{(r)}^a(\delta)\right]-\mathcal B^a(U)\\
&=\mathbb E_\gamma\left[\sum_{n=0}^K\chi_n(\delta)\Psi_n^a(\delta)\right]-\mathcal B^a(U),
\end{split}
\end{equation}
where
$$\mathcal B^a(U):=\sum_{j=1}^Ka_j\mathbb E_{Z_j}[U(\underline\delta_j(Z_j),Z_j)]\geq0.$$
 Nonnegativity follows from individual rationality and $a_j\geq0$.
\end{proposition}

\begin{proof}
Multiply \eqref{ConditionalGapRentIdentity} by $a_j$ and sum over $j$. Since $\sum_{j=1}^Ka_j=1$, substituting the resulting identity into \eqref{RankGapExpectedRevenue} yields the first representation in \eqref{GapRevenueIdentity}, with the residual defined above. The other two representations follow by rearranging the finite sums in \eqref{CumulativeRankAllocation}, \eqref{RankVirtualSurplus}, and \eqref{QuantityVirtualSurplus}.
\end{proof}

\noindent\textbf{The rank-gap virtual-surplus upper bound.} A quantity lottery averages these virtual surpluses, so pointwise maximization yields
$$B(a):=\mathbb E_\gamma\left[\max_{n=0,\dots,K}\Psi_n^a(\delta)\right].$$

\begin{proposition}[Rank-gap virtual-surplus upper bound]\label{GapVirtualSurplusUpperBound}
Under Assumption \ref{MightAsWelllWalkingOnTheSun}, every IC and IR top-rank menu $M$ satisfies $\mathrm{REV}(M)\leq B(a)$ for every $a\in\Delta_K$.
\end{proposition}

I next seek quantity prices that implement pointwise virtual-surplus maximization. Fix $a\in\Delta_K$ and a tie-breaking order over the finite quantity set. Let $n^a(\delta)\in\argmax_{m=0,\dots,K}\Psi_m^a(\delta)$ denote the pointwise virtual-surplus-maximizing quantity, with ties resolved by this order. For every $n=0,\dots,K$, the corresponding candidate quantity cell is $D_n^a:=\{\delta\in\Omega:n^a(\delta)=n\}$. By Lemma \ref{VirtualSurplusCellMeasurability} in the Appendix, each candidate cell is Borel measurable.

Only quantities selected on positive-probability cells are relevant for implementation. Define the candidate menu quantity set by
$$\mathcal N^a:=\{0\}\cup\left\{n\in\{1,\dots,K\}:\int_{D_n^a}\gamma(\delta)d\delta>0\right\}.$$
For $n,m\in\mathcal N^a$, let $b_{nm}^a:=\inf_{\delta\in D_n^a}\{V_n(\delta)-V_m(\delta)\}$ denote the infimum of the value advantage of quantity $n$ over quantity $m$ among types assigned to quantity $n$. I set $\inf\emptyset=+\infty$. I impose the following implementation condition.

\begin{condition}[Finite implementation of candidate virtual-surplus cells]\label{FiniteVirtualSurplusImplementation}
For the fixed weight vector $a$, there exist prices $\{p_n^a\}_{n\in\mathcal N^a}$ with $p_0^a=0$ and $p_n^a\geq0$ such that 
$p_n^a-p_m^a\leq b_{nm}^a$  for all $n,m\in\mathcal N^a$.
\end{condition}

Condition \ref{FiniteVirtualSurplusImplementation} is a finite quasilinear revealed-preference condition, related to \cite{Rochet1987} and \cite{Afriat1967}. 

\begin{remark}[Finite revealed-preference interpretation]\label{FiniteRPRemark}
The inequalities ensure that quantity $n$ maximizes payoff on $D_n^a$. The price-difference inequalities are feasible precisely when every finite cycle $n_1,\ldots,n_L$, with $n_{L+1}=n_1$, satisfies $\sum_{t=1}^L b_{n_tn_{t+1}}^a\geq0$. Subtracting the price of quantity zero from all prices preserves these inequalities and sets $p_0^a=0$. If $D_0^a\ne\varnothing$, then $-p_n^a\leq b_{0n}^a\leq0$ implies $p_n^a\geq0$. If $D_0^a=\varnothing$, nonnegativity must be checked separately. 
\end{remark}

Any price vector satisfying Condition \ref{FiniteVirtualSurplusImplementation} induces the candidate deterministic quantity-based menu $M^a:=\{(n,p_n^a):n\in\mathcal N^a\}$ with indirect utility $U^a(\delta):=\max_{n\in\mathcal N^a}\{V_n(\delta)-p_n^a\}$.
Condition \ref{FiniteVirtualSurplusImplementation} ensures that $M^a$ implements the candidate quantity rule $n^a$ almost surely, as shown in Lemma \ref{FiniteVirtualSurplusImplementationLemma} in the Appendix. 

\begin{condition}[Zero lower-face residual]\label{ZeroBoundaryResidual}
For the same fixed weight vector $a$ and the menu $M^a$ induced by the chosen price vector in Condition \ref{FiniteVirtualSurplusImplementation}, assume that $\mathcal B ^a(U^a)=0.$ 
\end{condition}

 Since utility and the weights are nonnegative, this condition requires utility to vanish almost surely on each conditional lower face with positive weight. 

 For the final claim of the next theorem, define $\mathcal A^*:=\argmin_{a\in\Delta_K}B(a)$; this set is nonempty and compact in the compact-support baseline (Lemma \ref{ConvexityOfB}).

\begin{theorem}[Finite implementation, residual gap, and exact optimality]\label{EndogenousGapFlowGap}\label{VWJetta}
Fix $a\in\Delta_K$ and suppose Assumptions \ref{ex} and \ref{MightAsWelllWalkingOnTheSun} hold. Suppose Condition \ref{FiniteVirtualSurplusImplementation} holds for this $a$, and fix one such price vector with induced menu $M^a$. Then
\begin{equation*}
0
\leq
\mathrm{REV}_{\mathcal Q}^{*}
-
\mathrm{REV}(M^a)
\leq
\mathcal B ^a(U^a).
\end{equation*}
If Condition \ref{ZeroBoundaryResidual} also holds for this same $a$ and $M^a$, then $M^a$ is optimal among all feasible general \textit{direct} mechanisms and
\begin{equation*}
\mathrm{REV}_{\mathcal Q}^{*}
=
\mathrm{REV}_{\mathrm{menu}}^{*}
=
\mathrm{REV}(M^a)
=
B(a).
\end{equation*}
In this case, moreover, we obtain the minimum revenue upper bound, i.e., $a\in\mathcal A^*$.
\end{theorem}

For any feasible top-rank menu $M$ with indirect utility $U$ and selected quantity lottery $\chi$, the revenue decomposition yields
\begin{equation}\label{SuoAiGeShanHai}
B(a)-\mathrm{REV}(M)=\mathbb E_\gamma\left[\max_{n=0,\dots,K}\Psi_n^a(\delta)-\sum_{n=0}^K\chi_n(\delta)\Psi_n^a(\delta)\right]+\mathcal B^a(U).
\end{equation}
Both terms on the right-hand side of \eqref{SuoAiGeShanHai} are nonnegative.  Implementing a virtual-surplus-maximizing quantity rule eliminates the first term, while zero utility on the active conditional lower faces eliminates the second. 

\noindent\textbf{Critical-price verification.} For each $c\in\mathcal C_\gamma$, summing its critical cutoffs from rank $r$ onward yields the rank price $\sum_{j=r}^Kc_j$. Summing these rank prices yields the candidate quantity prices
\begin{equation}\label{CriticalQuantityPrices}
P_n^c:=\sum_{r=1}^n\sum_{j=r}^Kc_j=V_n(c),\;\;n=0,\dots,K.
\end{equation}
At these prices, the payoff from quantity $n$ is $V_n(\delta)-V_n(c)$, so the critical comparison cells identify payoff-maximizing quantities. This payoff need not increase with quantity, since the differences $\delta_j-c_j$ can be negative. The associated menu is $M^c:=\{(0,0)\}\cup\{(n,P_n^c):n\in\mathcal N(c)\setminus\{0\}\}$, with indirect utility $U^c(\delta):=\max_{n\in\{0\}\cup\mathcal N(c)}\{V_n(\delta)-P_n^c\}$.

Lemma \ref{YouDontDeserveItImpliesZeroResidual} in the Appendix shows that Assumption \ref{YouDontDeserveIt} eliminates the lower-face residual of $M^c$.

\begin{proposition}[Closed-form exact optimality]\label{ClosedFormExactOptimality}
Suppose Assumptions \ref{ex} and \ref{MightAsWelllWalkingOnTheSun} hold, and there exists $c\in\mathcal C_\gamma$ satisfying both Assumptions \ref{TopRankSignAlignment} and \ref{YouDontDeserveIt}. For any such $c$, the associated menu $M^c$ is optimal among all feasible general \textit{direct} mechanisms. In particular,
\begin{equation*}
\mathrm{REV}_{\mathcal Q}^{*}
=
\mathrm{REV}_{\mathrm{menu}}^{*}
=
\mathrm{REV}(M^c).
\end{equation*}
Moreover, $a^c\in\mathcal A^*$.
\end{proposition}

\subsection{Extensions}\label{DroppingGapRentRegularityDiscussion} 
The main extension in this subsection allows conditional inverse-hazard ratios to depend on the other gaps. Its application to Example~\ref{StateDependentNonpurchaseExample} appears in Subsection~\ref{ApplicationStateDependent}.  Remarks~\ref{MultipleUnitsRevenueEquivalence} and~\ref{HomogeneousUnitsWithoutExchangeability} extend the framework to multiple units of heterogeneous products and to a homogeneous product, respectively; the latter requires no exchangeability. 

The conditional rent identity \eqref{StateDependentConditionalRentIdentity} permits direct verification of quantity prices when inverse-hazard ratios depend on the other gaps. Retain the interval sections in \eqref{ConditionalGapSection}, assume joint measurability and the integrability required for conditional integration by parts, and set undefined ratios to zero on the remaining $\gamma$-null set. Assumptions \ref{MightAsWelllWalkingOnTheSun}, \ref{TopRankSignAlignment}, and \ref{YouDontDeserveIt} are not imposed.
For $a\in\Delta_K$, define the state-dependent rank-gap virtual surplus by $\phi_j^{a,\mathrm{sd}}(\delta):=\delta_j-a_jR_j(\delta_j,Z_j)$, $j=1,\dots,K$. Define $\Phi_{(r)}^{a,\mathrm{sd}}$ and $\Psi_n^{a,\mathrm{sd}}$ by the same sums as in \eqref{RankVirtualSurplus} and \eqref{QuantityVirtualSurplus}, replacing $\phi_j^a(\delta_j)$ by $\phi_j^{a,\mathrm{sd}}(\delta)$, and set $\Psi_0^{a,\mathrm{sd}}=0$. Let $B^{\mathrm{sd}}(a):=\mathbb E_\gamma[\max_{n=0,\dots,K}\Psi_n^{a,\mathrm{sd}}(\delta)]$.
\begin{proposition}[State-dependent weight upper bound]\label{ConditionalGapFlowBound}
Under the conditional-section regularity maintained in this subsection, every IC and IR top-rank menu $M$ satisfies $\mathrm{REV}(M)\leq B^{\mathrm{sd}}(a)$ for every $a\in\Delta_K$.
\end{proposition}
Fix $a^{\mathrm{sd}}\in\Delta_K$ and a strict total order $\prec$ on $\{0,\dots,K\}$. Define
$$
D_n^{\mathrm{sd}}:=\left\{\delta\in\Omega:n=\min_{\prec}\left(\argmax_{m=0,\dots,K}\Psi_m^{a^{\mathrm{sd}},\mathrm{sd}}(\delta)\right)\right\},\;\; n=0,\dots,K,
$$
where $\min_{\prec}A$ denotes the first element of $A$ under $\prec$. Retain the outside option and all positive quantities assigned positive probability, so $\mathcal N^{\mathrm{sd}}:=\{0\}\cup\{n\in\{1,\dots,K\}:\int_{D_n^{\mathrm{sd}}}\gamma(\delta)d\delta>0\}$. For $n,m\in\mathcal N^{\mathrm{sd}}$, define $b_{nm}^{\mathrm{sd}}:=\inf_{\delta\in D_n^{\mathrm{sd}}}\{V_n(\delta)-V_m(\delta)\}$, with $\inf\emptyset=+\infty$. The following conditions extend Conditions \ref{FiniteVirtualSurplusImplementation} and \ref{ZeroBoundaryResidual} to state-dependent ratios.
\begin{condition}[State-dependent finite implementation]\label{StateDependentFiniteImplementation}
For the fixed weight vector $a^{\mathrm{sd}}$, there exist prices $\{p_n^{\mathrm{sd}}\}_{n\in\mathcal N^{\mathrm{sd}}}$ with $p_0^{\mathrm{sd}}=0$ and $p_n^{\mathrm{sd}}\geq0$ such that $p_n^{\mathrm{sd}}-p_m^{\mathrm{sd}}\leq b_{nm}^{\mathrm{sd}}$ for all $n,m\in\mathcal N^{\mathrm{sd}}$.
\end{condition}
Any price vector satisfying Condition \ref{StateDependentFiniteImplementation} induces the state-dependent candidate quantity-based menu $M^{\mathrm{sd}}:=\{(n,p_n^{\mathrm{sd}}):n\in\mathcal N^{\mathrm{sd}}\}$ with indirect utility $U^{\mathrm{sd}}(\delta):=\max_{n\in\mathcal N^{\mathrm{sd}}}\{V_n(\delta)-p_n^{\mathrm{sd}}\}$. All types face the same quantity prices; the conditional inverse-hazard ratios may vary with the other private gaps.
\begin{condition}[State-dependent zero lower-face residual]\label{StateDependentZeroBoundary}
For the same fixed weight vector $a^{\mathrm{sd}}$ and the menu $M^{\mathrm{sd}}$ induced by the chosen price vector in Condition \ref{StateDependentFiniteImplementation}, assume  $\mathcal B^{a^{\mathrm{sd}}}(U^{\mathrm{sd}})=0$.
\end{condition}
\begin{proposition}[Optimal menu under the state-dependent weight test]\label{StateDependentExactOptimality}
Fix $a^{\mathrm{sd}}\in\Delta_K$ and suppose Assumption \ref{ex} holds. Let $M^{\mathrm{sd}}$ be induced by a price vector satisfying Condition \ref{StateDependentFiniteImplementation} for this $a^{\mathrm{sd}}$. If Condition \ref{StateDependentZeroBoundary} holds for the same $a^{\mathrm{sd}}$ and $M^{\mathrm{sd}}$, and $B^{\mathrm{sd}}(a^{\mathrm{sd}})<\infty$, then $M^{\mathrm{sd}}$ is optimal among all feasible general \textit{direct} mechanisms. In particular,
$$
\mathrm{REV}_{\mathcal Q}^{*}=\mathrm{REV}_{\mathrm{menu}}^{*}=\mathrm{REV}(M^{\mathrm{sd}})=B^{\mathrm{sd}}(a^{\mathrm{sd}}).
$$
Moreover, $a^{\mathrm{sd}}\in\argmin_{a\in\Delta_K}B^{\mathrm{sd}}(a)$.
\end{proposition}

\begin{remark}[Multiple units of heterogeneous products]\label{MultipleUnitsRevenueEquivalence}
Partition $\mathcal K$ into nonempty groups $\mathcal K_1,\ldots,\mathcal K_N$, with $|\mathcal K_i|$ units of product $i$. Let $\theta$ follow the baseline distribution $F$ satisfying Assumption \ref{ex}. Within each product, order its coordinates from highest to lowest and interpret them as the buyer's private marginal values for successive units. Total value is the sum across products. Sorting within products preserves the globally ranked values, $\delta$, and $\gamma$. The highest $n$ marginal values form a feasible product-quantity choice: whenever a unit is selected, all earlier units of that product are also selected, with ties resolved in favor of earlier units. This choice attains $V_n(\delta)$. The Appendix proves that optimal expected revenue equals $\mathrm{REV}_{\mathcal Q}^{*}$. Whenever Theorem \ref{EndogenousGapFlowGap} establishes optimality of $M^a$, the same total-quantity prices are therefore optimal among all feasible IC and IR mechanisms, including lotteries over product quantities. 
\end{remark}
\begin{remark}[Homogeneous units without exchangeability]\label{HomogeneousUnitsWithoutExchangeability}
For a single homogeneous product, let $\theta^\downarrow$ denote the buyer's private vector of decreasing marginal values, with density $g$. No exchangeability assumption is required. Define $\delta$ as before, with density $\gamma(\delta)=g(\mathbf T\delta)$, and retain the density and conditional-section regularity assumptions in rank-gap coordinates. For any incentive-compatible mechanism with selected quantity lottery $\chi$ and indirect utility $U$, the envelope identity is $\partial_jU(\delta)=\sum_{n=0}^K\min\{n,j\}\chi_n(\delta)$ almost everywhere. The revenue decomposition therefore applies directly. Under Assumption \ref{MightAsWelllWalkingOnTheSun}, Conditions \ref{FiniteVirtualSurplusImplementation} and \ref{ZeroBoundaryResidual} ensure that $M^a$ is optimal among all feasible incentive-compatible and individually rational mechanisms, including quantity lotteries. Decreasing rank virtual surplus is not required.
\end{remark}

\section{Exact-optimality applications}\label{ExactOptimalityApplications}

This section applies Theorem \ref{EndogenousGapFlowGap} and Proposition \ref{ClosedFormExactOptimality} to the four distribution families. 

\subsection{Joint truncation: optimal menus with two positive quantities}\label{ApplicationJointTruncation}

For Distribution Family \ref{GaussianGapDensityFamily}, the following quantity prices attain the upper bound in Theorem \ref{EndogenousGapFlowGap}.

\begin{proposition}[Exact optimality under jointly truncated nonincreasing log-concave factors]\label{GaussianSparseExactness}
For Distribution Family \ref{GaussianGapDensityFamily}, consider the quantity prices
$$p_0^a=0,\;\; p_n^a=\sum_{j=2}^{K-1}j\overline\delta_j+\varpi+n\upsilon,\;\; n\in\{K-1,K\}.$$
Then $M^a=\{(0,0),(K-1,p_{K-1}^a),(K,p_K^a)\}$ is optimal among all feasible general \textit{direct} mechanisms. It satisfies Conditions \ref{FiniteVirtualSurplusImplementation} and \ref{ZeroBoundaryResidual}, and
$$\mathrm{REV}(M^a)=B(a),\;\; a\in\mathcal A^*.$$
Both positive quantities are selected with positive probability; quantity $0$ is selected with probability zero.
\end{proposition}

Under the remaining parameter restrictions, \eqref{GaussianLastGapBound} is necessary and sufficient for this menu and support to implement virtual-surplus maximization almost surely. At $c=(0,\dots,0,\upsilon)$, Assumption \ref{TopRankSignAlignment} holds, but Assumptions \ref{ActiveLocalInverseHazardAlignment} and \ref{YouDontDeserveIt} fail. The critical menu selects the same quantities at prices $(K-1)\upsilon$ and $K\upsilon$. Adding the common fee $\sum_{j=2}^{K-1}j\overline\delta_j+\varpi$ eliminates its positive boundary information rent. 

\noindent\textbf{Example~\ref{JointUniformLeadingExample} (continued).} Constant factors imply $R_2(z)=4-z$. Setting $\overline\delta_1=2$, $\overline\delta_2=4$, and $\varpi=1$ yields $\upsilon=2$ and $\ell_2=1$, with \eqref{GaussianLastGapBound} satisfied at equality. The support has area $17/2$, confirming the density $2/17$. Proposition~\ref{GaussianSparseExactness} establishes the optimality of prices $3$ for one object and $5$ for both. The buyer purchases one object with probability $2/17$ and both with probability $15/17$, so optimal expected revenue is $81/17$. Nonpurchase has probability zero. Figure~\ref{JointUniformLeadingFigure} illustrates these choices.

\begin{example}[A common-upper-bound specification: uniform and Gaussian cases]\label{GaussianCommonUpperSpecification}
Retain the density-factor restrictions of Distribution Family \ref{GaussianGapDensityFamily}. Fix $\overline\delta>0$, choose the terminal factor on $[0,\overline\delta]$, and let $\upsilon=R_K(\upsilon)$. Set $\overline\delta_j=\overline\delta$ for $2\leq j\leq K$, $\overline\delta_1=K\overline\delta-\upsilon$, and $\varpi=K(\overline\delta-\upsilon)$. All parameter restrictions hold, every actual projection is $I_j=[0,\overline\delta_j]$, and Assumption \ref{FlyMeToTheMoon} holds. The optimal menu is
$$
M^a=\left\{(0,0),
\left(K-1,\left[\frac{K(K+1)}2-1\right]\overline\delta-\upsilon\right),
\left(K,\left[\frac{K(K+1)}2-1\right]\overline\delta\right)\right\}.
$$

\noindent\textbf{Uniform factors.}
Constant factors imply $R_j(z)=\overline\delta_j-z$ and $\upsilon=\overline\delta/2$. For $K=3$ and $\overline\delta=1$, the optimal menu is $\{(0,0),(2,9/2),(3,5)\}$.

\noindent\textbf{Gaussian factors.}
For $K=3$ and $\overline\delta=1$, consider $\mathcal N(0,1)$ density factors, with their product jointly truncated to \eqref{GaussianGapSupport}. The cutoff solves $\upsilon=e^{\upsilon^2/2}\int_\upsilon^1e^{-s^2/2}ds$, so $\upsilon\approx0.4584638833$. The optimal menu is $\{(0,0),(2,5-\upsilon),(3,5)\}\approx\{(0,0),(2,4.5415361167),(3,5)\}$.
\end{example}

\subsection{A single active direction: full and sparse menus}\label{ApplicationSingleActive}

For the bounded distribution in Distribution Family \ref{SingleActiveGapDensityFamily}, the quantity prices are
\begin{equation}\label{SingleActiveQuantityPrices}
P_n^c=\min\{n,\iota\}\tau,\;\; n=0,\dots,K.
\end{equation}
Proposition \ref{ClosedFormExactOptimality} establishes exact optimality and $a^c\in\mathcal A^*$ for this bounded family. The support-relevant quantities are $\mathcal N(c)=\{0,1,\dots,\iota-1,K\}$, each selected with positive probability, as verified in the Appendix. When $\iota=K$, every quantity is offered at the common per-unit price $\tau$; when $\iota=1$, the menu reduces to pure bundling at price $\tau$. For $1<\iota<K$, quantities $1,\dots,\iota-1$ are priced at $n\tau$, and the full bundle is priced at $\iota\tau$.

The following proposition specializes Distribution Family \ref{SingleActiveGapDensityFamily} to $\iota=2$ and $\tau=\xi^{-1}$. The inactive factors are uniform, while the second factor has an exponential initial segment followed by a finite-endpoint tail.

\begin{proposition}[Rectangular product distribution with a sparse exact menu]\label{RectangularSparseExactness}
Let $K\geq3$ and $\xi>0$. Fix $\overline\delta_j>0$, $j=1,\dots,K$, such that
\begin{equation}\label{RectangularSparseBounds}
\overline\delta_2>\frac1\xi,\;\;
\sum_{\substack{j=1\\j\neq2}}^K\overline\delta_j
\leq
\frac1\xi,
\;\;\text{and}\;\;
\sum_{\substack{j=1\\j\neq2}}^K j\overline\delta_j
\leq
\frac2\xi.
\end{equation}
Retain the rectangular support of Distribution Family \ref{SingleActiveGapDensityFamily}, with $\iota=2$ and $\tau=\xi^{-1}$. For every $j\neq2$, let $\psi_j(z)=1/\overline\delta_j$ on $[0,\overline\delta_j]$. Choose the normalized second factor so that its upper mass satisfies
\begin{equation*}
S_2(z)=\begin{cases}
e^{-\xi z},&0\leq z\leq\xi^{-1},\\
e^{-1}\left(\dfrac{\overline\delta_2-z}{\overline\delta_2-\xi^{-1}}\right)^{\xi(\overline\delta_2-\xi^{-1})},&\xi^{-1}<z<\overline\delta_2.
\end{cases}
\end{equation*}
The second factor is therefore $\psi_2=-S_2'$, and the product-density form \eqref{PRoductDensitty} has $C=1$. Then Assumptions \ref{FlyMeToTheMoon}, \ref{ActiveLocalInverseHazardAlignment}, and \ref{YouDontDeserveIt} hold with $c=(0,1/\xi,0,\dots,0)$ and $a^c=(0,1,0,\dots,0)$. Moreover, $\mathcal N(c)=\{0,1,K\}$, and the finite top-rank menu
$$M^c=\left\{(0,0),\left(1,\frac1\xi\right),\left(K,\frac2\xi\right)\right\}$$
is exactly optimal among all feasible general \textit{direct} mechanisms.
\end{proposition}

The following lemma replaces the bounded active gap with an exponential distribution on $[0,\infty)$ while retaining the support bounds on the inactive gaps.

\begin{lemma}[Unbounded single-active extension]\label{UnboundedSingleActiveExactness}
Retain $K$, $\iota$, $\tau$, the bounded inactive intervals and their density factors from Distribution Family \ref{SingleActiveGapDensityFamily}, together with the two inactive-support inequalities in \eqref{SingleActiveDensitySupportBounds}. Replace the support and active factor with
$$\Omega=\{\delta\in\mathbb R_+^K:\delta_j\leq\overline\delta_j\text{ for every }j\neq\iota\},\;\;\psi_{\iota}(z)=\tau^{-1}e^{-z/\tau},\;z\geq0.$$
Let $\gamma$ have the product form \eqref{PRoductDensitty}, and let the object-value density be its exchangeable symmetrization on $\overline\Theta=\mathbb R_+^K$. The critical vector and weights in \eqref{SingleActiveCriticalCoordinates} generate an exactly optimal menu $M^c$ with prices \eqref{SingleActiveQuantityPrices}, $\mathcal N(c)=\{0,1,\dots,\iota-1,K\}$, and
$$\mathrm{REV}_{\mathcal Q}^*=\mathrm{REV}(M^c)=B(a^c)<\infty,\;\;a^c\in\mathcal A^*.$$
Every listed quantity is selected with positive probability. This extension does not impose the compact-support requirement of Assumption \ref{FlyMeToTheMoon}.
\end{lemma}

\noindent\textbf{Example~\ref{UnboundedLeadingExample} (continued).} This is the case $K=2$, $\iota=2$, and $\tau=1/2$ of Lemma~\ref{UnboundedSingleActiveExactness}, with $\overline\delta_1=1/2$. The inactive-support bounds hold, and $R_2(z)=1/2$ throughout the support. The critical values are therefore $c_1=0$ and $c_2=1/2$, yielding optimal quantity prices $P_1=1/2$ and $P_2=1$. Although the support extends to infinity along the $45^\circ$ direction, the probability of purchasing both objects is only $e^{-1}\approx36.79\%$; the probabilities of purchasing one object and purchasing none are $e^{-1/2}-e^{-1}\approx23.87\%$ and $1-e^{-1/2}\approx39.35\%$, respectively. Figure~\ref{UnboundedLeadingFigure} illustrates these choices.

\begin{example}[Exponential rank gaps]\label{ConditionalExponentialObjectValues}
Fix $K\geq2$, $\lambda>0$, and $\iota\in\{1,\ldots,K\}$. Let the rank gaps be independent, with $\delta_\iota\sim\mathrm{Exp}(\iota\lambda)$ and, for $j\neq\iota$, density $\psi_j(z)=j\lambda e^{-j\lambda z}/(1-e^{-j\lambda\overline\delta_j})$ on $[0,\overline\delta_j]$. Assign the resulting ranked values uniformly to object labels. Set $\tau=1/(\iota\lambda)$ and choose positive finite inactive bounds satisfying the two inactive-support inequalities in \eqref{SingleActiveDensitySupportBounds}. By Lemma \ref{UnboundedSingleActiveExactness}, the menu with prices $P_n=\min\{n,\iota\}/(\iota\lambda)$ is optimal among all feasible IC and IR mechanisms. Example~\ref{UnboundedLeadingExample} is the case $K=2$, $\iota=2$, $\lambda=1$, and $\overline\delta_1=1/2$.
\end{example}

\subsection{Global proportional alignment: quantity discounts}\label{ApplicationQuantityDiscounts}
The following sufficient condition is more general than the full-active family in Distribution Family \ref{FullActiveGapDensityFamily}: it requires proportional alignment only for gaps that enter comparisons among the support-relevant quantities.

\begin{proposition}[Global proportional alignment]\label{ActiveCommonCriticalScaling}
Suppose Assumptions \ref{ex} and \ref{MightAsWelllWalkingOnTheSun} hold. Fix $c\in\mathcal C_\gamma$, $\underline n\in\{0,\dots,K-1\}$, and $\kappa>0$ such that $\mathcal N(c)\subseteq\{\underline n,\dots,K\}$ and $\phi_j^{a^c}(z)=\kappa(z-c_j)$ for every $j>\underline n$ and every relevant $z$. Suppose also that $\Psi_{\underline n}^{a^c}(\delta)\geq\Psi_m^{a^c}(\delta)$ for every $m<\underline n$ and almost every $\delta\in\Omega$. Then Assumption \ref{TopRankSignAlignment} holds for this $c$. If Assumption \ref{YouDontDeserveIt} also holds for the same $c$, $M^c$ is exactly optimal. If $\underline n=0$ and $c\in\operatorname{int}\Omega$, every quantity is support relevant.
\end{proposition}
For $n,m\geq\underline n$, gaps $j\leq\underline n$ cancel from the comparison because $\min\{n,j\}=\min\{m,j\}=j$. Comparisons with smaller quantities follow from the assumed virtual dominance.

\begin{proposition}[Optimal quantity discounts]\label{CommonScalingExactOptimality}
For the distribution in Distribution Family \ref{FullActiveGapDensityFamily}, the menu $M^c$ is exactly optimal with
$$P_n^c=\frac{\overline\delta}{1+K\alpha}\frac{n(2K-n+1)}2,\;\;n=0,\dots,K.$$
Every quantity is selected with positive probability, and $a^c\in\mathcal A^*$.
\end{proposition}
These prices follow from \eqref{CriticalQuantityPrices} and Proposition \ref{ClosedFormExactOptimality}. Their successive differences are $\overline\delta(K-n+1)/(1+K\alpha)$, so the optimal prices exhibit strict quantity discounts.

\noindent\textbf{Example~\ref{ReverseBetaThreeSpaces} (continued).} Here $K=2$, $\overline\delta=1$, and $\alpha=1/4$ in Distribution Family~\ref{FullActiveGapDensityFamily}. Proposition~\ref{CommonScalingExactOptimality} establishes optimal quantity prices $P_1=4/3$ and $P_2=2$. Figure~\ref{ReverseBetaThreeSpacesFigure} illustrates the resulting quantity choices.

\subsection{Local inverse-hazard alignment: the optimal two-part tariff}\label{ApplicationTwoPartTariff}
Distribution Family \ref{TwoActiveGapDensityFamily} yields an exactly optimal two-part tariff \citep[see, e.g.,][]{Oi1971} for every finite $K$. Only two rent directions are active, although the buyer's private type remains full dimensional.

\begin{proposition}[Optimal two-part tariff for the two-active family]\label{TwoActiveCriticalTariff}
For the distribution in Distribution Family \ref{TwoActiveGapDensityFamily}, the menu $M^c$ is exactly optimal with
$$P_0^c=0,\;\;P_n^c=c_1+nc_K,\;\;n=1,\dots,K,$$
where $c$ is given by \eqref{TwoActiveCriticalCoordinates}. Every positive quantity is selected with positive probability, and $a^c=a\in\mathcal A^*$.
\end{proposition}
Optimality follows from Proposition~\ref{ClosedFormExactOptimality}. The fixed fee is $c_1$ and the per-unit price is $c_K$. 
Proposition \ref{LocalMarginalMatchingExactness} in the Appendix provides an extension in which the first-direction rent need only be bounded by a fixed fee $T$. That result allows a first factor without an exponential initial segment and applies Theorem \ref{EndogenousGapFlowGap} directly, without requiring $T$ to be a critical coordinate. The finite-dimensional exact result complements the large-product approximation for two-part tariffs by \cite{Armstrong1999}.

\subsection{State-dependent verification: a full-support example}\label{ApplicationStateDependent}

\noindent\textbf{Example~\ref{StateDependentNonpurchaseExample} (continued).} The optimal mechanism offers only the full bundle at price $1$. The gap support is $\Omega=\{\delta\in\mathbb R_+^2:\delta_1+\delta_2\leq1\}$. For $a^{\mathrm{sd}}=(0,1)$, the active conditional inverse-hazard ratio is
$$
R_2(z,\delta_1)=
\begin{cases}
1-z,&\delta_1+2z<1,\\
1-\delta_1-z,&\delta_1+2z\geq1.
\end{cases}
$$
The conditional sections are intervals and $0\leq R_2\leq1$. Below $\delta_1+2\delta_2=1$, both positive quantities have negative virtual surplus; above it, quantity two uniquely maximizes virtual surplus. With the tie-breaking order $0\prec2\prec1$, the retained quantities are $\mathcal N^{\mathrm{sd}}=\{0,2\}$, with $b_{20}^{\mathrm{sd}}=1$ and $b_{02}^{\mathrm{sd}}=-1$. Thus $p_0^{\mathrm{sd}}=0$ and $p_2^{\mathrm{sd}}=1$ satisfy Condition~\ref{StateDependentFiniteImplementation}. At the active lower boundary $\delta_2=0$, utility is $\max\{0,\delta_1-1\}=0$, satisfying Condition~\ref{StateDependentZeroBoundary}. Since $B^{\mathrm{sd}}(a^{\mathrm{sd}})\leq\mathbb E[V_2]\leq2$, Proposition~\ref{StateDependentExactOptimality} establishes optimality among all feasible incentive-compatible and individually rational mechanisms. Nonpurchase and purchase occur with probabilities $1/3$ and $2/3$, respectively, yielding optimal expected revenue $2/3$. Figure~\ref{StateDependentAllocationsFigure} illustrates these choices.

\section{Bundle values}\label{NonadditiveQuantityExtension}

This section derives a revenue decomposition for arbitrary lotteries over bundles without imposing ex ante symmetry. I first establish a general verification result and then provide conditions under which the critical-value construction yields optimal quantity prices.

\subsection{Bundle values and second differences}

Retain the object set $\mathcal K$. The buyer privately observes $v=(v_S)_{\varnothing\ne S\subseteq\mathcal K}$, where $v_S$ is his value for receiving bundle $S$. Let $\widetilde f$ denote the joint density of these $2^K-1$ values, with compact support $\operatorname{supp}(v)$. Assume nonnegative bundle values and normalize $v_\varnothing=0$. Impose free disposal: adding objects to a bundle cannot lower its value, so $v_S\leq v_T$ whenever $S\subseteq T\subseteq\mathcal K$. Consequently, the highest attainable value is weakly increasing in quantity.\footnote{Free disposal is imposed for convenience; Remark \ref{BundleFreeDisposalRemark} explains how the analysis extends without it.}

For $n=1,\ldots,K$, define the quantity value as the highest value among size-$n$ bundles, $\widetilde V_n:=\max_{|S|=n}v_S$, and set $\widetilde V_0:=0$. Define the first differences by $\widetilde\theta_r:=\widetilde V_r-\widetilde V_{r-1}$ for $r=1,\ldots,K$. For $K\geq2$, define the second-difference coordinates by
$$
\begin{aligned}
\widetilde\delta_1&:=2\widetilde V_1-\widetilde V_2,\\
\widetilde\delta_j&:=2\widetilde V_j-\widetilde V_{j-1}-\widetilde V_{j+1},\;\; j=2,\ldots,K-1,\\
\widetilde\delta_K&:=\widetilde V_K-\widetilde V_{K-1}.
\end{aligned}
$$
For $K=1$, set $\widetilde\delta_1:=\widetilde V_1$. In general, $\widetilde\delta_j=\widetilde\theta_j-\widetilde\theta_{j+1}$ for $j<K$ and $\widetilde\delta_K=\widetilde\theta_K$. Free disposal ensures $\widetilde\theta_r\geq0$. It does not require these increments to decrease with quantity, so $\widetilde\delta_j$ may be negative for $j<K$.

Let $\widetilde\delta:=(\widetilde\delta_1,\ldots,\widetilde\delta_K)$. Successive summation yields
\begin{equation}\label{QuanitityValueEspression}
    \widetilde\theta_r=\sum_{j=r}^K\widetilde\delta_j,\;\; \widetilde V_n(\widetilde\delta)=\sum_{r=1}^n\widetilde\theta_r=\sum_{j=1}^K\min\{n,j\}\widetilde\delta_j.
\end{equation}
The difference definitions and \eqref{QuanitityValueEspression} establish an invertible linear map between the quantity-value vector and $\widetilde\delta$. The maximum in $\widetilde V_n$ compares bundles of size $n$; the differences compare values across quantities. Conditional on the other coordinates, a feasible increase of $\widetilde\delta_j$ by $t$ increases $\widetilde V_n$ by $\min\{n,j\}t$, as in the additive case. These coefficients enter the envelope derivative and the quantity-price formula below.

For each nonempty bundle $S$, define its shortfall $\omega_S:=\widetilde V_{|S|}-v_S\geq0$ relative to the highest-valued bundle of the same size. Collect these differences in $\omega:=(\omega_S)_{\varnothing\ne S\subseteq\mathcal K}$. Equality between the values of two distinct bundles defines a hyperplane, which has probability zero under the joint density. Thus the maximizing size-$n$ bundle $S_n^\star$ is unique almost surely and satisfies $\omega_{S_n^\star}=0$. Rearranging the shortfall definition and substituting \eqref{QuanitityValueEspression} yields
\begin{equation}\label{BundleValueRecovery}
v_S=\widetilde V_{|S|}(\widetilde\delta)-\omega_S=\sum_{j=1}^K\min\{|S|,j\}\widetilde\delta_j-\omega_S,\;\; \varnothing\ne S\subseteq\mathcal K.
\end{equation}
On each region with maximizing bundles $(S_1^\star,\ldots,S_K^\star)$, the $K$ second differences and the $2^K-1-K$ remaining shortfalls form an invertible linear transformation of $v$, with inverse \eqref{BundleValueRecovery}. The zero shortfalls identify the maximizing bundles almost surely. These bundles need not be nested across quantities. 

Let $\Delta(2^{\mathcal K})$ denote the probability simplex over all subsets of $\mathcal K$, including the empty bundle. A feasible \textit{direct} mechanism $(\pi,p)$ consists of measurable maps $\pi:\operatorname{supp}(v)\to\Delta(2^{\mathcal K})$ and $p:\operatorname{supp}(v)\to\mathbb R_+$. After report $v$, the buyer receives bundle $S$ with probability $\pi_S(v)$ and pays $p(v)\geq0$. Allocation of the empty bundle means nonallocation. The buyer's information rent, or expected net utility from truthful reporting, is
$$
U(v):=\sum_{S\subseteq\mathcal K}\pi_S(v)v_S-p(v).
$$
The mechanism is \textit{incentive compatible} (IC) if
$$
U(v)\geq\sum_{S\subseteq\mathcal K}\pi_S(\hat v)v_S-p(\hat v)
$$
for all $v,\hat v\in\operatorname{supp}(v)$. With the outside option normalized to zero, it is \textit{individually rational} (IR) if $U(v)\geq0$ for all $v\in\operatorname{supp}(v)$. The seller's expected revenue is $\mathrm{REV}(\pi,p):=\mathbb E[p(v)]$.

Equation \eqref{BundleValueRecovery} inverts $v\mapsto(\widetilde\delta,\omega)$ on its image, so I write allocation and payment rules in these coordinates. Each report retains its allocation, payment, and payoff, preserving feasibility, IC, IR, and expected revenue. Expectations and conditional distributions are induced by $v$. Conditional distributions involving $\omega$ are regular conditional distributions; $\omega$ need not have a joint density. 

\subsection{Revenue decomposition and verification}

Fix $a\in\Delta_K$ and weight the rent identity for coordinate $j$ by $a_j$. A coordinate is active if $a_j>0$. Let $\widetilde\delta_{-j}$ collect the other second differences. For each active coordinate, assume that its support conditional on $(\widetilde\delta_{-j},\omega)$ is a nondegenerate interval with measurable endpoints and a positive conditional density. Denote its lower endpoint by $\widetilde\ell_j(\widetilde\delta_{-j},\omega)$ and its conditional density by $\gamma_j(\cdot\mid\widetilde\delta_{-j},\omega)$, extended by zero outside the interval. Define the conditional inverse-hazard ratio by
$$
R_j(z\mid\widetilde\delta_{-j},\omega):=\frac{\int_z^\infty\gamma_j(s\mid\widetilde\delta_{-j},\omega)\,ds}{\gamma_j(z\mid\widetilde\delta_{-j},\omega)}.
$$
The numerator is the conditional probability of values above $z$, and the denominator is the conditional density there. Utility at an endpoint is its limit from inside the conditional interval. 

For any feasible IC and IR mechanism, let $\chi_n:=\sum_{|S|=n}\pi_S$ denote the probability of allocating quantity $n$, including nonallocation when $n=0$. Conditional on $(\widetilde\delta_{-j},\omega)$, a feasible increase of $\widetilde\delta_j$ by $t$ raises every size-$n$ bundle value by $\min\{n,j\}t$. The envelope identity therefore yields
\begin{equation}\label{BundleEnvelope}
U_{\widetilde\delta_j}=\sum_{n=0}^K\min\{n,j\}\chi_n\;\;\text{almost surely}.
\end{equation}
This identity applies directly to lotteries over bundles.

For each quantity $n$, define quantity virtual surplus by
$$
\Psi_n^a(\widetilde\delta,\omega):=\widetilde V_n(\widetilde\delta)-\sum_{j:a_j>0}a_j\min\{n,j\}R_j(\widetilde\delta_j\mid\widetilde\delta_{-j},\omega),\;\; n=0,\ldots,K,
$$
with $\Psi_0^a=0$. Define the lower-boundary rent term as the weighted expected information rent at the active conditional lower endpoints:
$$
\mathcal B^a(U):=\sum_{j:a_j>0}a_j\mathbb E\!\left[U\bigl(\widetilde\ell_j(\widetilde\delta_{-j},\omega),\widetilde\delta_{-j},\omega\bigr)\right].
$$
Individual rationality implies $\mathcal B^a(U)\geq0$.

\begin{theorem}[Revenue decomposition and optimal quantity pricing]
\label{ConditionalBundleQuantityVerification}
Under the conditional-section regularity stated in this subsection, every feasible IC and IR mechanism satisfies
\begin{equation}\label{BundleValueRevenueDecomposition}
\mathrm{REV}(\pi,p)=\mathbb E\!\left[\sum_{n=0}^K\chi_n\Psi_n^a(\widetilde\delta,\omega)\right]-\mathbb E\!\left[\sum_{S\ne\varnothing}\pi_S\omega_S\right]-\mathcal B^a(U).
\end{equation}
Consequently,
$$
\mathrm{REV}(\pi,p)\leq\mathbb E\!\left[\max_{0\leq n\leq K}\Psi_n^a(\widetilde\delta,\omega)\right].
$$
Fix quantity prices $P_0=0$ and $P_n\geq0$, where every size-$n$ bundle has price $P_n$. The buyer's indirect utility under this menu, including the option of nonpurchase, is $U^P(\widetilde\delta):=\max_{0\leq n\leq K}\{\widetilde V_n(\widetilde\delta)-P_n\}$. Suppose that:

\textup{(i)} Almost surely,
$$
\argmax_{0\leq n\leq K}\{\widetilde V_n(\widetilde\delta)-P_n\}\subseteq\argmax_{0\leq n\leq K}\Psi_n^a(\widetilde\delta,\omega).
$$

\textup{(ii)} For each $j$ with $a_j>0$ and almost every conditioning state $(\widetilde\delta_{-j},\omega)$,
$$
U^P\bigl(\widetilde\ell_j(\widetilde\delta_{-j},\omega),\widetilde\delta_{-j}\bigr)=0.
$$

Then offering each bundle $S$ at price $P_{|S|}$ attains the revenue upper bound and is optimal among all feasible IC and IR mechanisms, including arbitrary lotteries over bundles.
\end{theorem}

The shortfall term in \eqref{BundleValueRevenueDecomposition} measures the expected loss from allocating a bundle below the highest value at its quantity. It is nonnegative and vanishes under a quantity-based menu because the buyer selects a maximizing bundle. In Theorem \ref{ConditionalBundleQuantityVerification}, conditions \textup{(i)} and \textup{(ii)} ensure quantity virtual-surplus maximization and zero lower-boundary rent, respectively. 

\subsection{Independence and separate critical pricing}

\begin{assumption}[Independence of highest values and bundle shortfalls]
\label{BundleValueIndependence}
The quantity-value vector is independent of the shortfall vector:
$$
(\widetilde V_1,\ldots,\widetilde V_K)\perp\omega. 
$$
\end{assumption} 

Since the quantity-value vector and $\widetilde\delta$ determine each other, Assumption \ref{BundleValueIndependence} is equivalent to $\widetilde\delta\perp\omega$. 

To obtain separate critical values, impose the following counterpart of Assumption \ref{MightAsWelllWalkingOnTheSun}. For each $j$ and almost every $\widetilde\delta_{-j}$, the conditional support of $\widetilde\delta_j$ is a nondegenerate interval with measurable endpoints and a positive conditional density. Denote its lower endpoint by $\widetilde\ell_j(\widetilde\delta_{-j})$ and its conditional density by $\gamma_j(\cdot\mid\widetilde\delta_{-j})$, extended by zero outside its conditional support. Require its conditional inverse-hazard ratio to depend only on the coordinate itself:
$$
\frac{\int_z^\infty\gamma_j(s\mid\widetilde\delta_{-j})\,ds}{\gamma_j(z\mid\widetilde\delta_{-j})}=R_j(z)
$$
almost everywhere on that interval, where $R_j(z)$ is the conditional inverse-hazard ratio and does not depend on the conditioning state. Assume that $R_j$ admits a continuous extension to the $j$th coordinate projection of $\operatorname{supp}(\widetilde\delta)$. Under Assumption \ref{BundleValueIndependence}, the same ratio and conditional lower endpoint apply after conditioning additionally on $\omega$. 

Fix $a\in\Delta_K$. For each active coordinate $j$, choose a nonnegative critical value $c_j$ in the interior of its coordinate projection satisfying $c_j=a_jR_j(c_j)$. Thus $c_j$ is a zero of the coordinate virtual-surplus function $z-a_jR_j(z)$. Set $c_j=0$ for inactive coordinates and construct the candidate quantity prices by
$$
P_n:=\sum_{j=1}^K\min\{n,j\}c_j,\;\; n=0,\ldots,K.
$$
The price $P_n$ applies to every size-$n$ bundle.

\begin{proposition}[Optimal quantity pricing for bundle values]
\label{NonadditiveQuantityVerification}
Suppose Assumption \ref{BundleValueIndependence} and the preceding gap rent regularity hold. For the same $a$, $c$, and $P$, suppose that:

\textup{(i)} For almost every supported $\widetilde\delta$, every $n\in\argmax_{0\leq r\leq K}\{\widetilde V_r-P_r\}$, and every $m=0,\ldots,K$,
\begin{equation}\label{NonadditiveRentComparison}
\sum_{j:a_j>0}[\min\{n,j\}-\min\{m,j\}]a_j[R_j(c_j)-R_j(\widetilde\delta_j)]\geq0.
\end{equation}

\textup{(ii)} For each $j$ with $a_j>0$ and almost every $\widetilde\delta_{-j}$, all quantities satisfy $\widetilde V_n\leq P_n$ when $\widetilde\delta_j=\widetilde\ell_j(\widetilde\delta_{-j})$.

Then the quantity-based menu with prices $P_n$ maximizes expected revenue among all feasible IC and IR mechanisms, including arbitrary lotteries over bundles.
\end{proposition}

Separate critical pricing also remains available without Assumption \ref{BundleValueIndependence} under the following conditional restriction. 

{
\renewcommand{\theassumption}{5$'$}
\renewcommand{\theHassumption}{5prime}
\begin{assumption}[Conditional inverse-hazard invariance]
\label{BundleConditionalInverseHazardInvariance}
Fix $a\in\Delta_K$ and maintain the conditional-section regularity of Theorem \ref{ConditionalBundleQuantityVerification}. For each $j$ with $a_j>0$, there exists a univariate function $R_j$ such that
$$
R_j(z\mid\widetilde\delta_{-j},\omega)=R_j(z)
$$
for almost every conditioning state and almost every $z$ in its conditional support.
\end{assumption}
}
\addtocounter{assumption}{-1}

\begin{corollary}[Critical-price verification]
\label{ConditionalBundleCriticalPricing}
Maintain the setting of Theorem \ref{ConditionalBundleQuantityVerification} and Assumption \ref{BundleConditionalInverseHazardInvariance}, with continuous versions of $R_j$. For each $j$ with $a_j>0$, choose a nonnegative root $c_j=a_jR_j(c_j)$ in the interior of the $j$th coordinate projection of $\operatorname{supp}(\widetilde\delta)$, and set $c_j=0$ otherwise. If the prices $P_n:=\sum_{j=1}^K\min\{n,j\}c_j$ satisfy conditions \textup{(i)} and \textup{(ii)} of Theorem \ref{ConditionalBundleQuantityVerification}, the quantity-based menu is optimal among all feasible IC and IR mechanisms, including arbitrary lotteries over bundles.
\end{corollary}

\begin{remark}[Without free disposal]\label{BundleFreeDisposalRemark}
The revenue decomposition and verification of bundle prices in Theorem \ref{ConditionalBundleQuantityVerification} remain valid without free disposal, retaining nonnegative bundle values and all other stated conditions. Quantity values may then decrease with quantity, and their first differences may be negative. The proof of Theorem \ref{ConditionalBundleQuantityVerification} bounds utilities, payments, shortfall losses, and maximum virtual surplus by $\max_{S\subseteq\mathcal K}v_S$, which is bounded under compact support. The verified menu offers every bundle $S$ at price $P_{|S|}$. With nondecreasing quantity prices, it also permits the buyer to select up to the purchased quantity, since any smaller bundle is available directly at a weakly lower price. The critical prices satisfy this property because $P_n-P_{n-1}=\sum_{j=n}^Kc_j\geq0$.
\end{remark}

Assumption \ref{BundleValueIndependence}, together with the gap rent regularity preceding Proposition \ref{NonadditiveQuantityVerification}, implies Assumption \ref{BundleConditionalInverseHazardInvariance}. The following jointly uniform example satisfies Assumption \ref{BundleConditionalInverseHazardInvariance}, and its quantity prices satisfy conditions \textup{(i)} and \textup{(ii)} of Theorem \ref{ConditionalBundleQuantityVerification}.

\begin{example}[Jointly uniform bundle values]\label{UniformBundleFreeDisposal}
Let $K=2$. The joint density of bundle values is
$$
\widetilde f(v)=\frac{16}{63}\,\mathbf1\!\left\{\begin{gathered}
0\leq\max\left\{\max\{v_{\{1\}},v_{\{2\}}\}-\frac32,\,\frac{v_{\{1,2\}}}{2}-\frac98\right\}\leq\min\{v_{\{1\}},v_{\{2\}}\},\\
2\max\{v_{\{1\}},v_{\{2\}}\}-\frac32\leq v_{\{1,2\}}\leq\max\{v_{\{1\}},v_{\{2\}}\}+\frac32
\end{gathered}\right\}.
$$
These restrictions imply nonnegative values and free disposal. Assumption \ref{BundleConditionalInverseHazardInvariance} holds with $a=(0,1)$ and $R_2(z)=3/2-z$. The quantity prices $P_1=3/2$ and $P_2=9/4$ satisfy the verification conditions of Theorem \ref{ConditionalBundleQuantityVerification} and therefore maximize expected revenue among all feasible IC and IR mechanisms, including lotteries over bundles. The buyer purchases one object with probability $3/14$ and both objects with probability $11/14$, yielding optimal expected revenue $117/56$.
\end{example}

Restricting Example~\ref{UniformBundleFreeDisposal} to $\widetilde\delta_1\geq3/4$ and renormalizing yields a jointly uniform distribution with density $32/81$ satisfying Assumption \ref{BundleValueIndependence}. Conditional on the quantity-value vector, either singleton is the maximizer with probability $1/2$, and the other singleton's shortfall is independently uniform on $[0,3/2]$. The optimal prices remain $P_1=3/2$ and $P_2=9/4$, with purchase probabilities $1/3$ and $2/3$ and expected revenue $2$.

\section{Concluding remark}
This study derives revenue upper bounds from second differences of quantity values and determines quantity prices that attain them under explicit distributional and support conditions. These menus are optimal among all IC and IR mechanisms. Depending on the distribution, they offer quantity discounts, a restricted set of quantities, or a two-part tariff. The results apply to ex ante symmetric additive values, nonadditive bundle values without label symmetry, and the multiple-unit environments described above.

\newpage
\appendix
\section*{Appendices}
\section{Preliminary and representation results}

\begin{lemma}[Rank-gap change of variables]\label{ChangeWeNeed}
For every Borel $A\subseteq\Omega$, $G(\mathbf T A)=\int_A g(\mathbf T\delta)\,d\delta$. Consequently, $\gamma(\delta)=g(\mathbf T\delta)$ almost everywhere.
\end{lemma}

\begin{proof}[Proof of Lemma \ref{ChangeWeNeed}]
Since $\theta^\downarrow=\mathbf T\delta$ and $\det\mathbf T=1$, the change-of-variables formula yields both claims.
\end{proof}

\begin{proof}[Proof of Lemma \ref{QRankEquivalence}]
Following \citet[Proposition 1 and Theorem 2]{BikhchandaniMishra2024}, average any feasible mechanism over object-label permutations. The requirement of permutation invariance of the support ensures that all relabeled reports are feasible. Taking the average over permutations preserves incentive compatibility, individual rationality, and nonnegative payments, while exchangeability preserves expected revenue. For the resulting symmetric mechanism, incentive compatibility against a report exchanging coordinates $i$ and $j$ implies $(\theta_i-\theta_j)(q_i(\theta)-q_j(\theta))\geq0$. Hence the allocation is rank preserving, and restriction to sorted types defines a feasible rank-allocation mechanism with the same expected revenue. Conversely, every feasible rank-allocation mechanism induces a feasible labeled mechanism with the same expected revenue. The two constructions establish equality of the revenue suprema.
\end{proof}

To characterize the buyer's choice under a quantity quota $\zeta\in[0,K]$, define the feasible allocation set $\mathcal Y(\zeta):=\{y\in[0,1]^K:\langle\mathbf 1,y\rangle\leq\zeta\}$ and the highest attainable value $V(\theta;\zeta):=\max_{y\in\mathcal Y(\zeta)}\langle y,\theta\rangle$. The following lemma characterizes the buyer's optimal choice.

\begin{lemma}[Top-rank choice under quota]\label{quantitycompress}
For any $\theta\in\Theta$ and $\zeta\in[0,K]$,
\begin{equation}\label{Vquota}
V(\theta;\zeta)=\sum_{r=1}^{\lfloor\zeta\rfloor}\theta_{(r)}+\mathbbm{1}\{\zeta\notin\mathbb Z\}(\zeta-\lfloor\zeta\rfloor)\theta_{(\lceil\zeta\rceil)}.
\end{equation}
The second term is omitted when $\zeta$ is an integer. A maximizer is $y^*(\theta;\zeta)=\Gamma(\theta)^\top y^\downarrow(\theta;\zeta)$, where
$$
y_{(r)}^\downarrow(\theta;\zeta)=
\begin{cases}
1,&r\leq\lfloor\zeta\rfloor,\\
\zeta-\lfloor\zeta\rfloor,&r=\lfloor\zeta\rfloor+1,\\
0,&r>\lfloor\zeta\rfloor+1.
\end{cases}
$$
\end{lemma}

\begin{proof}[Proof of Lemma \ref{quantitycompress}]
Fix $\theta$ and $\zeta$. Consider the co-sorted coordinates associated with $\theta^\downarrow$. The buyer's problem becomes
\begin{equation*}
\begin{split}
&\max_{y_{(r)}}\sum_{r=1}^{K}y_{(r)}\theta_{(r)}\\
\text{s.t.}\;\; &\sum_{r=1}^{K}y_{(r)}\leq\zeta,\;\;
0\leq y_{(r)}\leq1,\;\; r=1,\dots,K.
\end{split}
\end{equation*}
Let $\lambda\geq0$ be the multiplier for the quota constraint, let $\mu_{(r)}\geq0$ be the multiplier for the upper-bound constraint $y_{(r)}\leq1$, and let $\eta_{(r)}\geq0$ be the multiplier for the lower-bound constraint $y_{(r)}\geq0$. The Lagrangian is
\begin{equation*}
\mathcal L
=
\sum_{r=1}^{K}y_{(r)}\theta_{(r)}
+
\lambda\left(\zeta-\sum_{r=1}^{K}y_{(r)}\right)
+
\sum_{r=1}^{K}\mu_{(r)}(1-y_{(r)})
+
\sum_{r=1}^{K}\eta_{(r)} y_{(r)}.
\end{equation*}
The first-order condition with respect to $y_{(r)}$ is $\theta_{(r)}-\lambda-\mu_{(r)}+\eta_{(r)}=0$ for $r=1,\dots,K$. Complementary slackness yields $\lambda(\zeta-\sum_{r=1}^K y_{(r)})=0$, $\mu_{(r)}(1-y_{(r)})=0$, and $\eta_{(r)}y_{(r)}=0$ for every $r$.
Hence, if $y_{(r)}=1$, then $\theta_{(r)}\geq\lambda$; if $y_{(r)}=0$, then $\theta_{(r)}\leq\lambda$; and if $0<y_{(r)}<1$, then $\theta_{(r)}=\lambda$. Since $\theta_{(1)}\geq\theta_{(2)}\geq\cdots\geq\theta_{(K)}$, by the KKT conditions, there is an optimal solution that allocates to higher-ranked objects before lower-ranked objects.

The specified $y^\downarrow(\theta;\zeta)$ is therefore an optimizer. Substituting this optimizer into the objective yields the value in equation~\eqref{Vquota}. 
\end{proof}

\begin{proof}[Proof of Lemma \ref{RankMenuEquivalence}]
First, fix any $x\in\mathcal X$ and define $\chi$ by the reverse differences in the lemma. Monotonicity of $x$ yields $\chi_n\geq0$ for every $n$, and telescoping yields $\sum_{n=0}^{K}\chi_n=1$ and, for every $r=1,\dots,K$, $\sum_{n=r}^{K}\chi_n=x_{(r)}.$ Thus $x=\sum_{n=0}^{K}\chi_ne_n$. Conversely, every $e_n$ belongs to $\mathcal X$, and the vectors $e_0,\dots,e_K$ are affinely independent. Hence they are exactly the extreme points and the lottery representation is unique.

I next establish the menu representation applying the taxation principle. Every admissible report-independent top-rank menu induces an IC rank-allocation mechanism because the selected item maximizes the buyer's payoff among all menu items. If the menu includes the zero-quantity lottery at zero price, the induced mechanism is also IR. For labeled deviations, the rearrangement inequality implies that, for any reported sorted type $\tau$ and any true type $\theta$, the best labeling of the rank allocation $x(\tau)$ yields value $\langle x(\tau),\theta^\downarrow\rangle=V_{\chi(\tau)}(\theta)$. Hence no labeled deviation yields more than the corresponding top-rank menu item.

Conversely, fix any feasible direct rank-allocation mechanism $(x,p)\in\mathcal M_{\mathcal X}$. For every $\tau\in\Theta^\downarrow$, let $\chi(\tau)$ be the unique top-rank lottery representation of $x(\tau)$ and form the report-independent menu $M=\{(\chi(\tau),p(\tau)):\tau\in\Theta^\downarrow\}.$ For every $\theta\in\Theta$ and $\tau\in\Theta^\downarrow$, set $\theta':=\Gamma(\tau)^\top\theta^\downarrow\in\Theta$. Incentive compatibility at $\theta'$ against the report $\tau$ yields
\begin{equation*}
\begin{aligned}
\langle x(\theta^\downarrow),\theta^\downarrow\rangle-p(\theta^\downarrow)
&\geq\langle q^x(\tau),\theta'\rangle-p(\tau)\\
&=\langle x(\tau),\theta^\downarrow\rangle-p(\tau)
=V_{\chi(\tau)}(\theta)-p(\tau).
\end{aligned}
\end{equation*}
Thus truthful selection maximizes payoff from the menu, and individual rationality, payments, and revenue are preserved. The two constructions imply $\mathrm{REV}_{\mathcal X}^{*}=\mathrm{REV}_{\mathrm{menu}}^{*}$. Lemma \ref{QRankEquivalence} then yields the stated three-way revenue equality.
\end{proof}

\section{Rank-gap upper bound and exact verification}

\begin{proof}[Proof of Proposition \ref{GapVirtualSurplusUpperBound}]
By Proposition \ref{GapRevenueDecomposition}, revenue equals expected quantity virtual surplus minus the nonnegative residual. Since $\chi(\delta)$ is a probability vector, its weighted virtual surplus is at most $\max_n\Psi_n^a(\delta)$ at every type. Integrating this inequality yields $\mathrm{REV}(M)\leq B(a)$.
\end{proof}

\begin{lemma}[Measurability of candidate virtual-surplus cells]\label{VirtualSurplusCellMeasurability}
Under Assumption \ref{MightAsWelllWalkingOnTheSun}, for every $a\in\Delta_K$, each $D_n^a$ is Borel measurable.
\end{lemma}

\begin{proof}[Proof of Lemma \ref{VirtualSurplusCellMeasurability}]
By Assumption \ref{MightAsWelllWalkingOnTheSun}, each $R_j$ is continuous on $I_j$, so each $\Psi_n^a$ is continuous on $\Omega$. Under the fixed tie-breaking order, each $D_n^a$ is a finite intersection of sets defined by strict or weak comparisons between these continuous functions. Such sets are Borel measurable.
\end{proof}

\begin{lemma}[Finite implementation]\label{FiniteVirtualSurplusImplementationLemma}
Under Condition \ref{FiniteVirtualSurplusImplementation}, $M^a$ admits a selection that implements the candidate quantity rule $n^a$ almost surely. In particular, for every $n\in\mathcal N^a$,
\begin{equation*}
\delta\in D_n^a
\;\Longrightarrow\;
n\in
\argmax_{m\in\mathcal N^a}
\{V_m(\delta)-p_m^a\}.
\end{equation*}
\end{lemma}

\begin{proof}[Proof of Lemma \ref{FiniteVirtualSurplusImplementationLemma}]
Fix $n\in\mathcal N^a$ and $\delta\in D_n^a$. For every $m\in\mathcal N^a$, we have $V_n(\delta)-V_m(\delta)\geq b_{nm}^a\geq p_n^a-p_m^a.$ Hence $V_n(\delta)-p_n^a\geq V_m(\delta)-p_m^a.$ Thus quantity $n$ is payoff maximizing on $D_n^a$. Choose quantity $n$ on each retained candidate cell and the first payoff-maximizing menu item under a fixed order on the remaining types. Since there are finitely many quantities and the candidate cells are measurable, this selection implements $n^a$ almost surely.
\end{proof}

\begin{lemma}[Existence of a minimizing weight vector]\label{ConvexityOfB}
Suppose each state-independent ratio $R_j$ is nonnegative and Borel measurable, is finite $\gamma$-almost surely, and $\mathbb E_\gamma[V_K(\delta)]<\infty$. Then $B:\Delta_K\to\mathbb R$ is finite, continuous, and convex, and $\mathcal A^*$ is nonempty and compact. These requirements hold under Assumption \ref{MightAsWelllWalkingOnTheSun} in the compact-support baseline.
\end{lemma}

\begin{proof}[Proof of Lemma \ref{ConvexityOfB}]
For every $a\in\Delta_K$, nonnegativity of the terms $a_jR_j(\delta_j)$ and $\Psi_0^a=0$ yield 
$$0\leq\max_{n=0,\dots,K}\Psi_n^a(\delta)\leq\max_{n=0,\dots,K}V_n(\delta)=V_K(\delta).$$
The displayed bound implies $B(a)\leq\mathbb E_\gamma[V_K]<\infty$. Almost surely, $\max_n\Psi_n^a(\delta)$ is continuous and convex in $a$, being the maximum of finitely many affine functions. Dominated convergence, with dominating function $V_K$, establishes continuity of $B$, while expectation preserves convexity. Since $B$ is continuous on the compact simplex $\Delta_K$, the Weierstrass extreme value theorem ensures that it attains the  minimum. The minimizer set is closed in $\Delta_K$ and thus compact. 
\end{proof}

\begin{proof}[Proof of Theorem \ref{EndogenousGapFlowGap}]
By Proposition \ref{GapVirtualSurplusUpperBound} and Lemma \ref{RankMenuEquivalence}, we have $\mathrm{REV}_{\mathcal Q}^{*}=\mathrm{REV}_{\mathrm{menu}}^{*}\leq B(a).$ Lemma \ref{FiniteVirtualSurplusImplementationLemma} implies that $M^a$ implements the virtual-surplus-maximizing quantity rule generated by $a$. Proposition \ref{GapRevenueDecomposition} therefore yields $\mathrm{REV}(M^a)=B(a)-\mathcal B ^a(U^a).$ Combining these relations yields the upper bound on the revenue gap, while feasibility of $M^a$ yields the lower bound. If Condition \ref{ZeroBoundaryResidual} holds, then $\mathrm{REV}(M^a)=B(a)$, so feasibility and the upper bound imply $\mathrm{REV}_{\mathcal Q}^{*}=\mathrm{REV}_{\mathrm{menu}}^{*}=\mathrm{REV}(M^a)=B(a)$.

Finally, $B(\widetilde a)\geq\mathrm{REV}(M^a)=B(a)$ for every $\widetilde a\in\Delta_K$, so $a\in\mathcal A^*$.
\end{proof}

\begin{lemma}\label{YouDontDeserveItImpliesZeroResidual}
Assumption \ref{YouDontDeserveIt} implies $\mathcal B^{a^c}(U^c)=0$. 
\end{lemma}

\begin{proof}[Proof of Lemma \ref{YouDontDeserveItImpliesZeroResidual}]
Assumption \ref{YouDontDeserveIt} ensures that every positive item of $M^c$ yields nonpositive utility on each active conditional lower face. The outside option yields zero utility, so $U^c$ becomes zero on these faces. The definition of $\mathcal B^{a^c}$ yields the claim.
\end{proof}

\begin{proof}[Proof of Proposition \ref{ClosedFormExactOptimality}]
Fix any $c\in\mathcal C_\gamma$ satisfying Assumptions \ref{TopRankSignAlignment} and \ref{YouDontDeserveIt}. The prices in \eqref{CriticalQuantityPrices} implement the critical comparison cells. Assumption \ref{TopRankSignAlignment} ensures that a quantity maximizing $V_n(\delta)-V_n(c)$ also maximizes quantity virtual surplus almost surely, while Lemma \ref{YouDontDeserveItImpliesZeroResidual} eliminates the lower-face residual. Thus \eqref{SuoAiGeShanHai} yields $\mathrm{REV}(M^c)=B(a^c)$. Proposition \ref{GapVirtualSurplusUpperBound} and Lemma \ref{RankMenuEquivalence} establish exact optimality. Since every other weight vector also bounds the revenue of $M^c$, we have $a^c\in\mathcal A^*$.
\end{proof}

\section{Primitive conditions and application proofs}

\begin{lemma}[Product-density implication]\label{GapFactorizationSeparability}
Assumption \ref{FlyMeToTheMoon} implies Assumption \ref{MightAsWelllWalkingOnTheSun}. In particular, $R_j(z,Z_j)=R_j(z)=S_j(z)/\psi_j(z)$ for almost every $Z_j$ and every relevant $z$. Moreover, $R_j$ is nonincreasing and $R_j(z)\to0$ as $z\to\overline\delta_j^-$. 
\end{lemma}

\begin{proof}[Proof of Lemma \ref{GapFactorizationSeparability}]
Fix $j$ and condition on $Z_j$. By Assumption \ref{FlyMeToTheMoon},
\begin{equation}\label{RNoInCreaSing}
R_j(z,Z_j)
=
\frac{C\prod_{i\neq j}\psi_i(\delta_i)\int_z^{\overline\delta_j}\psi_j(s)ds}
{C\prod_{i\neq j}\psi_i(\delta_i)\psi_j(z)}
=
\frac{S_j(z)}{\psi_j(z)}
=
\frac1{h_j(z)}.
\end{equation}
Since $S_j$ is positive and continuously differentiable, log-concavity of $S_j$ implies that $(\log S_j)'=-h_j$ is nonincreasing. Hence $h_j$ is nondecreasing and $R_j=1/h_j$ is nonincreasing. Since $R_j$ is nonnegative, its left limit $L_j:=\lim_{z\to\overline\delta_j^-}R_j(z)$ exists and is finite. Suppose $L_j>0$. Then $R_j(z)\geq L_j$ and $\frac{S_j'(z)}{S_j(z)}=-\frac1{R_j(z)}\geq-\frac1{L_j}.$ Integrating from any $z_0<\overline\delta_j$ to $z$ yields $S_j(z)\geq S_j(z_0)\exp\left\{-\frac{z-z_0}{L_j}\right\}.$ Letting $z\uparrow\overline\delta_j$ yields $0=S_j(\overline\delta_j^-)\geq S_j(z_0)\exp\{-(\overline\delta_j-z_0)/L_j\}>0$, a contradiction. Hence $L_j=0$, and $R_j$ extends continuously to $\overline\delta_j$ with value zero. 
\end{proof}

\begin{remark}[Assumption \ref{FlyMeToTheMoon} in the original coordinates]\label{GapFactorizationOriginalCoordinates}
Assumption \ref{FlyMeToTheMoon} is a factorization in rank-gap coordinates, not in the original object coordinates. Under Assumption \ref{ex} and Lemma \ref{ChangeWeNeed}, for almost every $\theta$,
\begin{equation*}
f(\theta)=\frac{C}{K!}\mathbf 1_\Omega\left(\delta(\theta^\downarrow)\right)\prod_{j=1}^K\psi_j\left(\theta_{(j)}-\theta_{(j+1)}\right).
\end{equation*}
This is generally not an object-coordinate product density and does not reduce the type to a ray or a scalar taste. When all density factors are constant, the formula defines a constant density on $\Theta$, so the original object-value vector is jointly uniform on its restricted support. 
\end{remark}

\begin{lemma}[Existence of critical vectors]\label{PrimitiveCriticalVectorexistence}
Under Assumption \ref{FlyMeToTheMoon}, there exists a vector with every component positive in the critical-vector set $\mathcal C_\gamma$. Moreover, for every $i=1,\dots,K$, the critical-vector set $\mathcal C_\gamma$ contains a vector satisfying $a_i^c=1$ and $a_j^c=0$ for every $j\neq i$.
\end{lemma}

\begin{proof}[Proof of Lemma \ref{PrimitiveCriticalVectorexistence}]
By Lemma \ref{GapFactorizationSeparability}, let
$H_j(z):=\frac{z}{R_j(z)}=\frac{z\psi_j(z)}{S_j(z)}$.
Each $H_j$ is continuous on $(\ell_j,\overline\delta_j)$, Assumption \ref{FlyMeToTheMoon}-\textup{(iii)} ensures $\sum_jH_j(\ell_j)<1$, and $H_j(z)\to+\infty$ as $z\to\overline\delta_j^-$. Choose each $c_j\in\operatorname{int}(I_j)$ sufficiently close to $\ell_j$ that $\sum_jH_j(c_j)<1$. Fixing all but one coordinate, increase the remaining coordinate toward its upper endpoint until $\sum_jH_j(c_j)=1$; existence follows from continuity and divergence at the upper endpoint. The resulting vector belongs to $\mathcal C_\gamma$ and has every component positive.  Each $c_j$ lies in the interior of its coordinate interval $I_j$, but the vector $c$ need not lie in $\operatorname{int}\Omega$.

For a fixed $i$, $H_i(\ell_i)<1$ and $H_i(z)\to+\infty$ as $z\to\overline\delta_i^-$. Hence there exists $c_i\in(\ell_i,\overline\delta_i)$ such that $H_i(c_i)=1$. Setting $c_j=0$ for every $j\neq i$ yields $a_i^c=1$ and $a_j^c=0$ for every $j\neq i$, so the resulting vector belongs to $\mathcal C_\gamma$.
\end{proof}

\begin{lemma}[Local inverse-hazard alignment]\label{ActiveLocalInverseHazardAlignmentImpliesAlignment}
For a fixed $c\in\mathcal C_\gamma$, Assumption \ref{ActiveLocalInverseHazardAlignment} implies Assumption \ref{TopRankSignAlignment} for the same $c$.
\end{lemma}

\begin{proof}[Proof of Lemma \ref{ActiveLocalInverseHazardAlignmentImpliesAlignment}]
Fix $c\in\mathcal C_\gamma$ satisfying Assumption \ref{ActiveLocalInverseHazardAlignment}. Part \textup{(i)} of Assumption \ref{ActiveLocalInverseHazardAlignment} establishes virtual-surplus maximization on $\mathscr C_0(c)$. Fix $n\in\mathcal N(c)\setminus\{0\}$ and almost every $\delta\in\mathscr C_n(c)$, and apply \eqref{CriticalComparisonIdentity}. If $m<n$, each coefficient in the final sum is nonnegative. Part \textup{(ii)} ensures nonnegative inverse-hazard differences for active $j\leq n$, while part \textup{(iii)} sets these differences to zero for active $j>n$. If $m>n$, the coefficients become zero for $j\leq n$, and part \textup{(iii)} eliminates all remaining terms. In either case, the final sum is nonnegative, and $\delta\in\mathscr C_n(c)$ implies $V_n(\delta)-V_n(c)\geq V_m(\delta)-V_m(c)$. Thus $\Psi_n^{a^c}(\delta)\geq\Psi_m^{a^c}(\delta)$ for every $m$, establishing Assumption \ref{TopRankSignAlignment}.
\end{proof}

\subsection*{Verification of the distribution families}

\begin{proof}[Verification of Distribution Family \ref{GaussianGapDensityFamily} and proof of Proposition \ref{GaussianSparseExactness}]
For each density factor,
$$
\frac{S_j(z)}{\psi_j(z)}
=\int_0^{\overline\delta_j-z}
\frac{\psi_j(z+t)}{\psi_j(z)}dt
\leq\overline\delta_j-z.
$$
The inequality follows from monotonicity of $\psi_j$. Log-concavity of $\psi_j$ implies that the integrand is nonincreasing in $z$, while the integration interval strictly shortens. Thus $S_j/\psi_j$ is continuous and strictly decreasing, and becomes zero at its upper endpoint. Consequently, $\upsilon=R_K(\upsilon)$ has a unique solution in $(0,\overline\delta_K)$.

The parameter bounds imply $0\leq\ell_K<\upsilon$. The support is coordinatewise increasing within its ambient box, has nonempty path-connected interior, and is the closure of its interior. Maximizing over the other coordinates identifies $I_K=[\ell_K,\overline\delta_K]$. Every nonempty conditional section is a single interval with upper endpoint $\overline\delta_j$. The product-density calculation in \eqref{RNoInCreaSing} therefore yields $R_j(z,Z_j)=S_j(z)/\psi_j(z)$ on the relevant sections, establishing Assumption \ref{MightAsWelllWalkingOnTheSun}. Parts \textup{(i)} and \textup{(ii)} of Assumption \ref{FlyMeToTheMoon} also hold; part \textup{(iii)} need not hold for the general parameter specification.

The payoff difference between quantities $K$ and $K-1$ is $\delta_K-\upsilon$. If $\delta_K\leq\upsilon$, support membership implies
$$
(\overline\delta_1-\delta_1)
+\sum_{j=2}^{K-1}j(\overline\delta_j-\delta_j)
\leq(K-1)(\delta_K-\ell_K).
$$
The left side is at least $(K-1)(\overline\delta_{K-1}-\delta_{K-1})$, including when $K=2$. Hence
$$
\delta_{K-1}+\delta_K
\geq\overline\delta_{K-1}+\ell_K
\geq R_K(\ell_K)
\geq R_K(\delta_K).
$$
Since $a_K=1$ and the other weights are zero, this implies $\Phi_{(K-1)}^a\geq0$ and $\Phi_{(r)}^a\geq\Phi_{(K-1)}^a$ for every $r<K-1$. Quantity $K-1$ therefore dominates every smaller quantity in virtual surplus, including quantity zero. If $\delta_K\geq\upsilon$, then $\phi_K^a(\delta_K)\geq0$ and every rank virtual surplus is nonnegative, so quantity $K$ maximizes virtual surplus. Moreover, $\Psi_K^a-\Psi_{K-1}^a=\phi_K^a(\delta_K)$ has the same sign as $\delta_K-\upsilon$. This establishes Condition \ref{FiniteVirtualSurplusImplementation}.

The identity $p_{K-1}^a=\overline\delta_1+\sum_{j=2}^{K-1}j\overline\delta_j+(K-1)\ell_K$ shows that both positive-quantity payoffs are nonpositive on the ambient face $\delta_K=0$. At a positive conditional lower endpoint, the defining support inequality binds by continuity; at a zero endpoint, the same equality follows from this bound and support membership. Thus $U^a(\underline\delta_K(Z_K),Z_K)=0$ for almost every $Z_K$. Since only the terminal weight is positive, Condition \ref{ZeroBoundaryResidual} holds. Theorem \ref{EndogenousGapFlowGap} establishes exact optimality and $a\in\mathcal A^*$.

Choosing the first $K-1$ gaps sufficiently close to their upper endpoints and $\delta_K$ slightly below or above $\upsilon$ yields open sets selecting $K-1$ or $K$, respectively, with positive utility. Both quantities therefore have positive selection probability. The zero-utility boundary has measure zero, so quantity zero has selection probability zero.

For necessity of \eqref{GaussianLastGapBound}, consider $(\overline\delta_1,\dots,\overline\delta_{K-1},\ell_K)$. Quantity $K-1$ has zero payoff and strictly exceeds quantity $K$ there. If the condition fails, then $\Phi_{(K-1)}^a=\overline\delta_{K-1}+\ell_K-R_K(\ell_K)<0$. Increasing the terminal gap slightly and decreasing the other gaps by sufficiently small amounts yields an open set selecting $K-1$ with positive utility but satisfying $\Psi_{K-2}^a>\Psi_{K-1}^a$. This argument also applies to $K=2$, where $\Psi_{K-2}^a=\Psi_0^a=0$.

To compare this menu with the critical-price construction, set $c=(0,\dots,0,\upsilon)$, so $a^c=a$. When $\delta_K\leq\upsilon$, the preceding inequalities imply $\theta_{(K-1)}\geq R_K(\ell_K)>\upsilon$, so every critical rank margin through $K-1$ is positive. When $\delta_K\geq\upsilon$, every critical rank margin is nonnegative. Hence $\mathcal N(c)=\{K-1,K\}$, and the critical menu induces the same quantity choices as $M^a$, establishing Assumption \ref{TopRankSignAlignment}. Since $p_{K-1}^a\geq(K-1)R_K(\ell_K)>(K-1)\upsilon$, the common fee $\sum_{j=2}^{K-1}j\overline\delta_j+\varpi$ is positive. Thus $U^c=U^a+\sum_{j=2}^{K-1}j\overline\delta_j+\varpi$ on $\Omega$, and Assumption \ref{YouDontDeserveIt} fails at $c$. The positive-probability cell for $K-1$ contains an open set with $\delta_K<\upsilon$, where $R_K(\delta_K)>R_K(\upsilon)$, violating part \textup{(iii)} of Assumption \ref{ActiveLocalInverseHazardAlignment}.

Finally, fix any $c\in\mathcal C_\gamma$. For $j<K$, $c_j\leq a_j^c\overline\delta_j$, strictly whenever $a_j^c>0$. Likewise, $c_K\leq a_K^cR_K(\ell_K)$, strictly whenever $a_K^c>0$, because every positive $c_K$ lies above $\ell_K$ and $R_K$ is strictly decreasing. Since the weights sum to one,
$$V_{K-1}(c)<(1-a_K^c)\left(\overline\delta_1+\sum_{j=2}^{K-1}j\overline\delta_j\right)+(K-1)a_K^cR_K(\ell_K)\leq p_{K-1}^a.$$
The last inequality follows from \eqref{GaussianLastGapBound} and $p_{K-1}^a=\overline\delta_1+\sum_{j=2}^{K-1}j\overline\delta_j+(K-1)\ell_K$. The critical equation also implies $c_K\leq\upsilon$, so $V_K(c)<p_K^a$.
Consequently, $U^c(\delta)\geq U^a(\delta)+\min_{n\in\{K-1,K\}}\{p_n^a-V_n(c)\}>0$ on $\Omega$. Discarding quantities with zero-probability critical cells leaves this maximum unchanged, by continuity and $\Omega=\operatorname{cl}(\operatorname{int}\Omega)$. Since $U^a\geq0$, continuity at the conditional endpoints implies $\mathcal B^{a^c}(U^c)\geq\min_{n\in\{K-1,K\}}\{p_n^a-V_n(c)\}>0$. Thus Assumption \ref{YouDontDeserveIt} fails for every critical vector.
\end{proof}

\begin{proof}[Verification of Example \ref{GaussianCommonUpperSpecification}]
Since $R_K(z)\leq\overline\delta-z$, the cutoff satisfies $0<\upsilon\leq\overline\delta/2$. The specified parameters imply $\varpi=\overline\delta_1-(K-1)\upsilon>0$, $\varpi<\overline\delta_1$, and $\ell_K=0$. Moreover, $\overline\delta_{K-1}\geq\overline\delta\geq R_K(0)$, so \eqref{GaussianLastGapBound} holds.

Setting the first gap to zero and all other gaps to their upper endpoints results in zero payoff from quantity $K$. Setting the terminal gap to zero and all other gaps to their upper endpoints results in zero payoff from quantity $K-1$. For an intermediate gap $j$, setting $\delta_j=0$ and all other gaps to their upper endpoints implies $V_K-p_K^a=(K-j)\overline\delta-\upsilon>0$. Thus every actual projection is $I_j=[0,\overline\delta_j]$. Assumption \ref{FlyMeToTheMoon} holds, including part \textup{(iii)}, whose left side is zero. Substitution into the price formula yields the stated menu. The critical-menu comparison and boundary-rent statements follow from the preceding proof.
\end{proof}

\begin{proof}[Verification of Distribution Family \ref{SingleActiveGapDensityFamily}]
The rectangular support and the stated density factors satisfy parts \textup{(i)} and \textup{(ii)} of Assumption \ref{FlyMeToTheMoon}; part \textup{(iii)} holds because all lower endpoints are zero. Since $S_{\iota}(0)=\tau\psi_{\iota}(0)$, integration of \eqref{SingleActiveDensityInitialSegment} yields $S_{\iota}(z)=\tau\psi_{\iota}(0)e^{-z/\tau}$ and $R_{\iota}(z)=\tau$ for $0\leq z\leq\tau$.
Log-concavity of $S_{\iota}$ implies that $R_{\iota}$ is nonincreasing, so $R_{\iota}(z)\leq\tau$ for $z\geq\tau$. Thus the vector in \eqref{SingleActiveCriticalCoordinates} belongs to $\mathcal C_\gamma$ and has the stated weights. Its quantity prices satisfy $V_n(c)=\min\{n,\iota\}\tau$. Consequently, quantities $\iota,\dots,K-1$, when present, have the same price as the full bundle and are strictly dominated almost surely.

Since $R_\iota(z)\leq R_\iota(\tau)=\tau$ throughout the support, part \textup{(ii)} of Assumption \ref{ActiveLocalInverseHazardAlignment} holds. If $\delta_\iota>\tau$, the buyer strictly prefers quantity $K$ to every quantity below $\iota$. Hence nonpurchase and support-relevant quantities $n<\iota$ occur only when $\delta_\iota\leq\tau$, where $R_\iota(\delta_\iota)=\tau$ and $\Psi_n^{a^c}(\delta)=V_n(\delta)-V_n(c)$ for every $n$. This establishes parts \textup{(i)} and \textup{(iii)}.

On the active lower face $\delta_\iota=0$, \eqref{SingleActiveDensitySupportBounds} implies $V_n(0,Z_\iota)\leq n\sum_{j\ne\iota}\overline\delta_j\leq n\tau$ for $n\leq\iota$ and $V_n(0,Z_\iota)\leq\sum_{j\ne\iota}j\overline\delta_j\leq\iota\tau$ for $n\geq\iota$. Thus $V_n(0,Z_\iota)\leq V_n(c)$ for every $n$, establishing Assumption \ref{YouDontDeserveIt}.

Types sufficiently close to the origin choose quantity zero, and interior types with $\delta_{\iota}>\tau$ choose quantity $K$. For $1\leq n<\iota$, set $\delta_{\iota}=\tau-\varepsilon$, $\delta_n=2\varepsilon$, and every other gap equal to $\varepsilon^2$. For sufficiently small $\varepsilon>0$, the first $n$ real marginal payoffs are positive, those from ranks $n+1$ through $\iota$ are negative of order $\varepsilon$, and the total value of ranks $\iota+1,\dots,K$ is of order $\varepsilon^2$. Quantity $n$ is therefore the unique maximizer. These are interior types, and strict choice persists on a neighborhood of each type. Positivity of the density proves $\mathcal N(c)=\{0,1,\dots,\iota-1,K\}$, with every listed quantity selected with positive probability.

 The normalized active factor specified in Proposition \ref{RectangularSparseExactness}, with $\xi=\tau^{-1}$ and its second coordinate relabeled as $\iota$, provides such a factor for every $\overline\delta_{\iota}>\tau$. The inactive factors may be uniform on intervals satisfying \eqref{SingleActiveDensitySupportBounds}.
\end{proof}

\begin{proof}[Proof of Lemma \ref{UnboundedSingleActiveExactness}]
The product density has $R_{\iota}(z)=\tau$ on $[0,\infty)$. For $j\neq\iota$, the factorization and log-concavity argument in Lemma \ref{GapFactorizationSeparability} yields $0\leq R_j(z)\leq R_j(0)<\infty$. Moreover, $\mathbb E_\gamma[V_K(\delta)]=\iota\tau+\sum_{j\neq\iota}j\mathbb E_\gamma[\delta_j]\leq\iota\tau+\sum_{j\neq\iota}j\overline\delta_j<\infty$.
For every feasible top-rank menu, $0\leq U,p_M\leq V_K$, so utility and payments are integrable. By Fubini's theorem, for almost every $Z_\iota$, \eqref{Envelope} holds for almost every $z$. Absolute continuity and integration by parts on $[0,t]$ yield
$$\int_0^t U(z,Z_\iota)\frac{e^{-z/\tau}}{\tau}\,dz=U(0,Z_\iota)-U(t,Z_\iota)e^{-t/\tau}+\int_0^t Q_\iota(z,Z_\iota)e^{-z/\tau}\,dz.$$
The boundary term converges to zero because $0\leq U(t,Z_\iota)e^{-t/\tau}\leq(\iota t+\sum_{j\neq\iota}j\overline\delta_j)e^{-t/\tau}\to0$. Both integrals have finite limits, since $U(z,Z_\iota)\leq\iota z+\sum_{j\neq\iota}j\overline\delta_j$ and $0\leq Q_\iota\leq\iota$. Monotone convergence therefore yields \eqref{StateDependentConditionalRentIdentity} for the active gap.

For $j\neq\iota$, finite-endpoint integration by parts applies, with $0\leq Q_jR_j\leq jR_j(0)$; also $0\leq Q_\iota R_\iota\leq\iota\tau$. Lower-face utilities satisfy $0\leq U(0,Z_j)\leq\sum_{i\neq j}i\delta_i$ and have finite expectations. Integration over the conditioning states therefore establishes \eqref{ConditionalGapRentIdentity} for every $j$. Averaging yields the revenue decomposition and upper bound for all $a\in\Delta_K$, and Lemma \ref{ConvexityOfB} applies.

The rank-allocation reduction also extends to this support. Payments satisfy $0\leq p(\theta)\leq\sum_k\theta_k$ and are integrable. Averaging over the finitely many object-label permutations preserves IC, IR, and nonnegative payments, and preserves expected revenue by exchangeability. IC against coordinate transpositions ensures rank preservation. Sorting is Borel measurable, so the remaining arguments in Lemmas \ref{QRankEquivalence} and \ref{RankMenuEquivalence} apply. Thus every $a\in\Delta_K$ bounds revenue over general direct mechanisms.

For the vector in \eqref{SingleActiveCriticalCoordinates}, $\Psi_n^{a^c}(\delta)=V_n(\delta)-\min\{n,\iota\}\tau$ throughout the support. The prices in \eqref{SingleActiveQuantityPrices} therefore implement pointwise virtual-surplus maximization. Quantities $\iota,\dots,K-1$, when present, are dominated by the full bundle at the same price. The inactive support bounds yield $V_n(0,Z_{\iota})\leq\min\{n,\iota\}\tau$ for every $n$, so the candidate utility becomes zero on the only active lower face. The revenue identity therefore yields $\mathrm{REV}(M^c)=B(a^c)$, proving exact optimality. For every $a\in\Delta_K$, we have $B(a)\geq\mathrm{REV}(M^c)=B(a^c)$. Thus $a^c\in\mathcal A^*$. The interior types identified in the bounded-family verification remain available, proving the stated support-relevant quantity set and positive-probability selection.
\end{proof}

\begin{proof}[Verification of Distribution Family \ref{FullActiveGapDensityFamily}]
The rectangular support and the stated factors satisfy Assumption \ref{FlyMeToTheMoon}: the upper masses are log-concave and every lower endpoint is zero. Direct integration yields $R_j(z)=(\overline\delta-z)/\alpha$. The coordinates and weights stated in the distribution family satisfy \eqref{ZeroCrossingVS}, and
\begin{equation}\label{FullActiveVirtualScaling}
\phi_j^{a^c}(z)=\frac{1+K\alpha}{K\alpha}(z-c_j),\;\;j=1,\dots,K.
\end{equation}
The quantity sums in \eqref{QuantityVirtualSurplus} then establish Assumption \ref{TopRankSignAlignment}. On a lower face, monotonicity yields $V_n(0,Z_j)\leq\overline\delta\{\sum_{i=1}^K\min\{n,i\}-\min\{n,j\}\}$. For $n\geq1$, the shape restriction yields
$$K\alpha\left\{\sum_{i=1}^K\min\{n,i\}-\min\{n,j\}\right\}\leq1\leq\min\{n,j\},$$
which yields $V_n(0,Z_j)\leq V_n(c)$ and proves Assumption \ref{YouDontDeserveIt}. Since $c$ is interior, small variations of the adjacent real critical margins ensure that each quantity uniquely maximizes payoff on an open set. Hence $\mathcal N(c)=\{0,\dots,K\}$.
\end{proof}

\begin{proof}[Verification of Distribution Family \ref{TwoActiveGapDensityFamily}]
The endpoint matching equations identify the two positive coordinates in \eqref{TwoActiveCriticalCoordinates} as $c_1=(K\beta_1-\beta_K)/(K-1)$ and $c_K=(\beta_K-\beta_1)/(K-1)$.
From the definitions of $\beta_1,\beta_K$, the two inequalities in \eqref{TwoActiveDensitySupportBounds} imply $\ell_1<c_1<\overline\delta_1$ and
$$\ell_K+\sum_{j=2}^{K-1}\overline\delta_j<c_K<\overline\delta_K.$$
In particular, the initial segments in \eqref{TwoActiveDensityInitialSegments} lie within the density-factor intervals. Integrating \eqref{TwoActiveDensityInitialSegments} yields $R_i(z)=R_i(\ell_i)$ for $\ell_i\leq z\leq c_i$, $i\in\{1,K\}$. Hence $a_iR_i(c_i)=c_i$. Log-concavity implies that each active rank-gap virtual surplus is strictly increasing, so these are its unique zero crossings.

The expressions for $c_1$ and $c_K$ above imply that the threshold for quantity $n$ in \eqref{TwoActivePrimitiveSupport} is $V_n(c)=c_1+nc_K$. The support is coordinatewise increasing within its rectangle, its upper corner satisfies the defining inequalities strictly, and every nonempty conditional section ends at $\overline\delta_j$. The points with $\delta_1=\ell_1$ and all other coordinates at their upper endpoints, and with $\delta_K=\ell_K$ and all other coordinates at their upper endpoints, belong to $\Omega$. The support bounds also imply $\overline\delta_K-\ell_K>\sum_{j=2}^{K-1}\overline\delta_j$, so setting all intermediate coordinates to zero and the first and terminal coordinates to their upper endpoints yields a feasible point. Thus the stated intervals are the coordinate projections. The support has path-connected nonempty interior and is the closure of its interior. The product form yields \eqref{RNoInCreaSing}, and
$$\sum_{j=1}^K\frac{\ell_j}{R_j(\ell_j)}=\frac{\ell_1}{R_1(\ell_1)}+\frac{\ell_K}{R_K(\ell_K)}<a_1+a_K=1.$$
Assumption \ref{FlyMeToTheMoon} follows.

The exclusion cell has zero probability, since it is contained in the finite union of the hyperplanes $V_n(\delta)=V_n(c)$. Each active $R_i$ is constant up to $c_i$ and nonincreasing thereafter, establishing part \textup{(ii)} of Assumption \ref{ActiveLocalInverseHazardAlignment}. If $\delta_K>c_K$, then for every $1\leq n<K$,
$$\{V_K(\delta)-V_K(c)\}-\{V_n(\delta)-V_n(c)\}=\sum_{j=n+1}^{K-1}(j-n)\delta_j+(K-n)(\delta_K-c_K)>0.$$
Consequently, every positive critical type cell with $n<K$ lies in $\delta_K\leq c_K$, where $R_K$ is constant. This proves part \textup{(iii)} and hence local alignment.

On the ambient face $\delta_1=\ell_1$, quantity $K$ has zero net value at the upper corner, and every other positive quantity has nonpositive net value. On the ambient face $\delta_K=\ell_K$, quantity $1$ has zero net value at the upper corner; the displayed bound on $c_K$ ensures negative subsequent marginal net values. Monotonicity extends these inequalities over the respective faces. At an actual conditional lower endpoint above its ambient lower bound, the defining support condition binds by continuity. Thus all positive quantities have nonpositive critical net value on both active conditional lower faces, proving Assumption \ref{YouDontDeserveIt}.

For sufficiently small $\varepsilon>0$, quantity $1$ is uniquely selected near $(c_1+2\varepsilon,0,\dots,0,c_K-\varepsilon)$. Each $2\leq n<K$ is uniquely selected near the vector with $\delta_1=c_1+\varepsilon$, $\delta_n=2\varepsilon$, $\delta_K=c_K-\varepsilon$, and the other intermediate coordinates zero. Quantity $K$ is uniquely selected near $(c_1+\varepsilon,0,\dots,0,c_K+\varepsilon)$. We can increase the zero coordinates slightly to obtain interior types. Net value remains positive, and the same quantity is still strictly preferred. Hence, $\mathcal N(c)=\{1,\dots,K\}$.
\end{proof}

\subsection*{Verification of the application results}

\begin{proof}[Proof of Proposition \ref{RectangularSparseExactness}]
The stated density factors are positive, continuous, and integrable on their respective half-open intervals. For every $j\neq2$, uniformity yields $R_j(z)=\overline\delta_j-z$. Direct differentiation of $S_2$ yields
$$R_2(z)=\begin{cases}
\xi^{-1},&0\leq z\leq\xi^{-1},\\  
\dfrac{\overline\delta_2-z}{\xi(\overline\delta_2-\xi^{-1})},&\xi^{-1}<z<\overline\delta_2.
\end{cases}$$
These inverse-hazard ratios are nonincreasing, so all upper masses are log-concave. Moreover, $S_2(0)/\psi_2(0)=\xi^{-1}$ and $\psi_2(z)=\xi e^{-\xi z}$ for $0\leq z\leq\xi^{-1}$. Thus the distribution satisfies the factor restrictions of Distribution Family \ref{SingleActiveGapDensityFamily} with $\iota=2$ and $\tau=\xi^{-1}$, while \eqref{RectangularSparseBounds} yields its support restrictions.  Proposition \ref{ClosedFormExactOptimality} therefore establishes exact optimality of the stated menu.
\end{proof}

\begin{proof}[Proof of Proposition \ref{ActiveCommonCriticalScaling}]
Fix $n\in\mathcal N(c)$ and $\delta\in\mathscr C_n(c)$. For every $m\geq\underline n$, the terms with $j\leq\underline n$ cancel because $\min\{n,j\}=\min\{m,j\}=j$. Applying the assumed proportionality to the remaining terms yields
\begin{equation*}
\Psi_n^{a^c}(\delta)-\Psi_m^{a^c}(\delta)
=
\kappa\left(
\{V_n(\delta)-V_n(c)\}
-
\{V_m(\delta)-V_m(c)\}
\right)
\geq0.
\end{equation*}
If $m<\underline n$, applying this identity to $n$ and $\underline n$ together with the assumed virtual dominance yields $\Psi_n^{a^c}(\delta)\geq\Psi_{\underline n}^{a^c}(\delta)\geq\Psi_m^{a^c}(\delta)$. Thus Assumption \ref{TopRankSignAlignment} holds for this $c$. If Assumption \ref{YouDontDeserveIt} also holds for the same $c$, Proposition \ref{ClosedFormExactOptimality} establishes exact optimality.
If $\underline n=0$ and $c\in\operatorname{int}\Omega$, the triangular map from $\delta-c$ to the adjacent real critical margins $\left(\sum_{j=k}^K(\delta_j-c_j)\right)_{k=1}^K$ is invertible. Hence, near $c$, one can choose the adjacent margins positive for $k\leq n$ and negative for $k>n$, so each quantity $n=0,\dots,K$ uniquely maximizes payoff on an open set. Since $\gamma$ is positive almost everywhere on $\operatorname{int}\Omega$, every such open set has positive probability. Therefore every quantity is support relevant.
\end{proof}

\noindent\textbf{A two-part-tariff extension.}
The following extension bounds first-gap rent by the fixed fee without requiring the first density factor to have an exponential initial segment. 

\begin{proposition}[Local inverse-hazard alignment with a menu-aligned support]\label{LocalMarginalMatchingExactness}
Let $K\geq2$ and fix $\overline\delta_j>0$ for $j=1,\dots,K$. For $i\in\{1,K\}$, let $\psi_i$ be a positive, continuously differentiable, integrable density factor on $[0,\overline\delta_i)$ with nondecreasing hazard $h_i$. Suppose $c_K\in(0,\overline\delta_K)$ is the right endpoint of the maximal initial interval on which $h_K$ is constant, so that $h_K(z)=h_K(c_K)$ for $0\leq z\leq c_K$ and $h_K(z)>h_K(c_K)$ for $c_K<z<\overline\delta_K$. Set $a_K=c_Kh_K(c_K)$, $a_1=1-a_K$, and $a_j=0$ for $2\leq j<K$. Suppose $0<a_K<1$, and let $T\in(0,\overline\delta_1)$ satisfy $a_1/h_1(0)\leq T$ and
\begin{equation}\label{ParametersinTwoactives}
\overline\delta_1-T+\sum_{j=2}^{K-1}\overline\delta_j\leq c_K,\;\;\text{and}\;\;\sum_{j=2}^{K-1}j\overline\delta_j+K(\overline\delta_K-c_K)\leq T.
\end{equation}
For $2\leq j<K$, let $\psi_j$ be a positive constant density factor on $[0,\overline\delta_j]$. Let $\gamma$ have the form \eqref{PRoductDensitty} on
$$\Omega=\left\{\delta\in\prod_{j=1}^K[0,\overline\delta_j]:\max_{n=1,\dots,K}\{V_n(\delta)-T-nc_K\}\geq0\right\},$$
and let the object-value density be its exchangeable symmetrization. Then Assumption \ref{MightAsWelllWalkingOnTheSun} and Conditions \ref{FiniteVirtualSurplusImplementation} and \ref{ZeroBoundaryResidual} hold for $a$ and the prices $p_0^a=0$ and $p_n^a=T+nc_K$, $n=1,\dots,K$.
Every positive quantity is selected with positive probability. The resulting menu $M^a$ is optimal among all feasible general \textit{direct} mechanisms, and $a\in\mathcal A^*$.
\end{proposition}

\begin{proof}[Proof of Proposition \ref{LocalMarginalMatchingExactness}]
The support is coordinatewise increasing within the ambient rectangle, so every nonempty conditional section has the fixed upper endpoint $\overline\delta_j$. The factorization argument in \eqref{RNoInCreaSing} yields $R_i=1/h_i$ for $i\in\{1,K\}$ and $R_j(z)=\overline\delta_j-z$ for the intermediate coordinates. The finite upper endpoints and the stated density conditions yield bounded continuous extensions, establishing Assumption \ref{MightAsWelllWalkingOnTheSun}.

The first-direction condition implies $a_1R_1(z)\leq T$. Also, $a_KR_K(z)=c_K$ for $z\leq c_K$ and $a_KR_K(z)<c_K$ otherwise. Hence, for $n\geq1$,
$$\Psi_n^a(\delta)=\{V_n(\delta)-p_n^a\}+\{T-a_1R_1(\delta_1)\}+n\{c_K-a_KR_K(\delta_K)\}.$$
If $\delta_K\leq c_K$, the last term becomes zero and the middle term is nonnegative and common to all positive quantities. The support ensures a nonnegative maximum positive-quantity payoff, so every payoff-maximizing positive quantity also maximizes virtual surplus. If $\delta_K>c_K$, quantity $K$ strictly dominates every smaller positive quantity in both comparisons, and its virtual surplus is strictly positive. Thus the stated prices implement the virtual-surplus maximizers.

For sufficiently small $\varepsilon>0$, quantity $1$ is uniquely selected near $(T+2\varepsilon,0,\dots,0,c_K-\varepsilon)$. Each $2\leq n<K$ is uniquely selected near the vector with $\delta_1=T+\varepsilon$, $\delta_n=2\varepsilon$, $\delta_K=c_K-\varepsilon$, and all other intermediate coordinates zero. Quantity $K$ is uniquely selected near $(T+\varepsilon,0,\dots,0,c_K+\varepsilon)$. Small positive perturbations of the zero coordinates preserve these strict comparisons and yield interior types with positive payoffs. Hence every positive quantity has positive probability.

By \eqref{ParametersinTwoactives}, all positive-quantity payoffs are nonpositive on each of the ambient faces $\delta_1=0$ and $\delta_K=0$. If an actual conditional lower endpoint is positive, continuity and the defining support inequality imply that the maximum positive-quantity payoff is exactly zero there. Thus $U^a=0$ on both active conditional lower faces. Theorem \ref{EndogenousGapFlowGap} yields exact optimality and $a\in\mathcal A^*$.
\end{proof}

The first factor may, for example, be constant, in which case its rent condition reduces to $a_1\overline\delta_1\leq T$. The fixed fee is also recoverable from the lower endpoint $\ell_1$ of $I_1$:
$$T=\min_{\delta\in\Omega}\max_{n=1,\dots,K}\{V_n(\delta)-nc_K\}=\ell_1+\sum_{j=2}^{K-1}j\overline\delta_j+K(\overline\delta_K-c_K).$$
The first equality follows from the support definition, with equality at $(T,0,\dots,0,c_K)$. For the second, the smallest feasible first gap occurs when all other coordinates are at their upper endpoints, where quantity $K$ is the maximizing positive quantity. Neither equality requires $T=a_1R_1(T)$.

\section{Proofs for the extensions}

\begin{proof}[Proof of Theorem \ref{ConditionalBundleQuantityVerification}]
Fix a feasible IC and IR mechanism and an index $j$ with $a_j>0$. Conditional on $(\widetilde\delta_{-j},\omega)$, each menu payoff is affine in $\widetilde\delta_j$, with slope in $[0,j]$. IC represents truthful utility as the supremum of these payoffs. Hence utility has a finite convex Lipschitz extension along the conditional section and is absolutely continuous there. At every differentiability point, the selected payoff has the same derivative as utility. The positive conditional density therefore establishes \eqref{BundleEnvelope}. 
Conditional integration yields
$$
\begin{aligned}
\mathbb E[U\mid\widetilde\delta_{-j},\omega]
&=U\bigl(\widetilde\ell_j(\widetilde\delta_{-j},\omega),\widetilde\delta_{-j},\omega\bigr)\\
&\;+\mathbb E\!\left[R_j(\widetilde\delta_j\mid\widetilde\delta_{-j},\omega)U_{\widetilde\delta_j}\mid\widetilde\delta_{-j},\omega\right].
\end{aligned}
$$
The conditional mean of $R_j$ equals the integral of the conditional survival function and is bounded by the length of the conditional section. Compact support bounds these lengths uniformly. Together with $0\leq U_{\widetilde\delta_j}\leq j$, this establishes integrability of the rent terms. Nonnegative bundle values and payments, together with IR, imply $0\leq U,p\leq\max_{S\subseteq\mathcal K}v_S$ and $0\leq\sum_{S\ne\varnothing}\pi_S\omega_S\leq\max_{S\subseteq\mathcal K}v_S$. Conditional endpoint utilities satisfy the same bound by continuity.

Integrating the $j$th rent identity, multiplying by $a_j$, and summing over the active directions yields
$$
\mathbb E[U]=\mathcal B^a(U)+\mathbb E\!\left[\sum_{j:a_j>0}a_jR_j(\widetilde\delta_j\mid\widetilde\delta_{-j},\omega)\sum_{n=0}^K\min\{n,j\}\chi_n\right].
$$
Substituting this expression and $\sum_S\pi_Sv_S=\sum_n\chi_n\widetilde V_n-\sum_{S\ne\varnothing}\pi_S\omega_S$ into expected revenue proves \eqref{BundleValueRevenueDecomposition}. Both subtracted terms are nonnegative. Since $\chi$ is a probability vector, $\sum_n\chi_n\Psi_n^a\leq\max_n\Psi_n^a$. Moreover, $0\leq\max_n\Psi_n^a\leq\max_{S\subseteq\mathcal K}v_S$ because $\Psi_0^a=0$ and $\Psi_n^a\leq\widetilde V_n\leq\max_{S\subseteq\mathcal K}v_S$ for every $n$. This proves the revenue bound. 

Under the quantity-based menu, every selected bundle has zero shortfall. Condition \textup{(i)} of Theorem \ref{ConditionalBundleQuantityVerification} ensures that the selected quantity maximizes $\Psi_n^a$, and condition \textup{(ii)} ensures $\mathcal B^a(U^P)=0$. Offering every bundle $S$ at price $P_{|S|}$, together with nonallocation at zero payment, defines an IC and IR mechanism that attains the bound.
\end{proof}

\begin{proof}[Proof of Proposition \ref{NonadditiveQuantityVerification}]
The quantity-value vector and $\widetilde\delta$ determine each other through the difference definitions and \eqref{QuanitityValueEspression}, so Assumption \ref{BundleValueIndependence} implies $\widetilde\delta\perp\omega$. Conditioning additionally on $\omega$ therefore leaves the conditional distribution of $\widetilde\delta_j$ given $\widetilde\delta_{-j}$ unchanged. The conditional inverse-hazard ratios in Theorem \ref{ConditionalBundleQuantityVerification} consequently equal $R_j(\widetilde\delta_j)$.

The critical-value equations imply
$$
\Psi_n^a=(\widetilde V_n-P_n)+\sum_{j:a_j>0}a_j\min\{n,j\}[R_j(c_j)-R_j(\widetilde\delta_j)].
$$
Condition \textup{(i)} of Proposition \ref{NonadditiveQuantityVerification} therefore ensures that every payoff-maximizing quantity also maximizes $\Psi_n^a$, and condition \textup{(ii)} ensures zero utility at every active conditional lower endpoint. Theorem \ref{ConditionalBundleQuantityVerification} then establishes optimality.
\end{proof}

\begin{proof}[Verification of Example \ref{UniformBundleFreeDisposal}]
On each singleton-ranking region, the transformation to second-difference coordinates and the nonzero singleton shortfall has unit Jacobian. The support implies $\widetilde\delta_1\in[-3/4,3/2]$ and $\widetilde\delta_2\in[\min\{3/2-\widetilde\delta_1,(9/4-\widetilde\delta_1)/2\},3/2]$; the nonzero shortfall lies between zero and $\min\{3/2,9/8+\widetilde\delta_1/2\}$. By the change-of-variables formula and Fubini's theorem, the total volume is $63/16$. The region $\widetilde\delta_2<3/4$, where one object is preferred to two, has volume $27/32$. Normalizing by the total volume establishes the stated density and purchase probabilities. If the highest singleton value is less than $3/2$, the purchase restriction implies $v_{\{1,2\}}\geq9/4$; otherwise the second restriction implies $v_{\{1,2\}}\geq2\widetilde V_1-3/2\geq\widetilde V_1$. Hence free disposal holds.

Conditional on $\widetilde\delta_1$ and all shortfalls, $R_2(z)=3/2-z$. For $a=(0,1)$, the payoff and virtual-surplus differences between quantities two and one are $\widetilde\delta_2-3/4$ and $2\widetilde\delta_2-3/2$. At the conditional lower endpoint, utility is zero and the selected positive quantity's virtual surplus is $3/2-\widetilde\delta_1\geq0$. Both positive quantities' virtual surpluses increase with $\widetilde\delta_2$, so Theorem \ref{ConditionalBundleQuantityVerification} applies. The prices and probabilities yield revenue $117/56$.
\end{proof}

\begin{proof}[Proof of Proposition \ref{ConditionalGapFlowBound}]
By the law of iterated expectations, integrating over $Z_j$ in \eqref{StateDependentConditionalRentIdentity} yields the corresponding unconditional rent identity. Multiplying the $j$th identity by $a_j$ and summing over $j$ yields
$$\mathbb E_\gamma[U(\delta)]=\mathbb E_\gamma\left[\sum_{j=1}^K a_jQ_j(\delta)R_j(\delta_j,Z_j)\right]+\mathcal B^a(U).$$
Substitution into \eqref{RankGapExpectedRevenue}, together with \eqref{CumulativeRankAllocation}, yields
\begin{equation}\label{StateDependentRevenueIdentity}
\mathrm{REV}(M)=\mathbb E_\gamma\left[\sum_{n=0}^K\chi_n(\delta)\Psi_n^{a,\mathrm{sd}}(\delta)\right]-\mathcal B^a(U).
\end{equation}
Since $\chi$ is a probability vector and $\mathcal B^a(U)\geq0$, pointwise maximization establishes the bound.
\end{proof}

\begin{proof}[Proof of Proposition \ref{StateDependentExactOptimality}]
Condition \ref{StateDependentFiniteImplementation} implements the candidate quantity rule by the same argument as in Lemma \ref{FiniteVirtualSurplusImplementationLemma}. Hence \eqref{StateDependentRevenueIdentity} yields $\mathrm{REV}(M^{\mathrm{sd}})=B^{\mathrm{sd}}(a^{\mathrm{sd}})-\mathcal B^{a^{\mathrm{sd}}}(U^{\mathrm{sd}})$. Condition \ref{StateDependentZeroBoundary} eliminates the residual. Proposition \ref{ConditionalGapFlowBound} and Lemma \ref{RankMenuEquivalence} therefore imply the stated revenue equalities. Every other weight vector provides an upper bound on the revenue of this same feasible menu, so $a^{\mathrm{sd}}$ minimizes $B^{\mathrm{sd}}$ over $\Delta_K$.
\end{proof}

\begin{proof}[Proof of Remark \ref{MultipleUnitsRevenueEquivalence}]
By the taxation principle, represent any feasible multiple-unit mechanism as a menu $M=\{(x^\nu,p^\nu):\nu\in\mathcal I\}$, where $x^\nu_{\mathcal K_i}$ consists of the nonincreasing unit-allocation probabilities for product $i$. In the original environment, offer 
$$
M^{\mathrm{lab}}
:=
\left\{
(\rho x^\nu,p^\nu):
\nu\in\mathcal I,\;
\rho\in\Sigma_{\mathcal K},\;
\rho(\mathcal K_i)=\mathcal K_i\ \forall i
\right\}.
$$
The rearrangement inequality yields
$$
\max_{\substack{\rho\in\Sigma_{\mathcal K}\\
\rho(\mathcal K_i)=\mathcal K_i\ \forall i}}
\langle\theta,\rho x^\nu\rangle
=
\sum_{i=1}^N
\left\langle
\theta_{\mathcal K_i}^{\downarrow},
x^\nu_{\mathcal K_i}
\right\rangle,
$$
where $\theta_{\mathcal K_i}^{\downarrow}$ denotes the within-product ordering. Retain the multiple-unit mechanism's menu selection and resolve ties among maximizing within-group permutations by a fixed rule. The resulting labeled mechanism is IC and IR and preserves selected payments and expected revenue. Optimal multiple-unit revenue is therefore at most $\mathrm{REV}_{\mathcal Q}^{*}$. 

Conversely, the globally highest $n$ marginal values form a feasible product-quantity allocation, with ties resolved in favor of preceding units. Its value is $V_n(\delta)$. For each top-rank quantity lottery, offer at the same price every feasible assignment of each possible realized quantity to products. The buyer can select an assignment attaining $V_n(\delta)$ at every realized quantity, so the highest expected value of the option is $\sum_{n=0}^K\chi_nV_n(\delta)$ in both environments. This report-independent menu therefore preserves menu comparisons and selected payments. Lemma \ref{RankMenuEquivalence} establishes the reverse revenue inequality, and Theorem \ref{EndogenousGapFlowGap} establishes the optimality claim for $M^a$.
\end{proof}

\bibliographystyle{te}
\bibliography{Bibliography}

\end{document}